\documentclass[journal,twoside,web]{ieeecolor}
\usepackage{generic}
\usepackage{cite}
\usepackage{textcomp}

\usepackage{amssymb,latexsym,amsfonts,amsmath,bbm}
\usepackage[cal=cm,bb=pazo]{mathalfa}
\usepackage{mathrsfs}
\usepackage{graphicx}
\usepackage{paralist}
\usepackage{dsfont}
\usepackage{eso-pic}
\usepackage{epsfig}
\usepackage{mathtools}
\usepackage{epstopdf}
\usepackage{lipsum}
\usepackage{cite}
\usepackage{subcaption}
\usepackage{booktabs}
\usepackage{multirow}

\usepackage[nocomma]{optidef}

\usepackage{tikz}
\usetikzlibrary{calc,shapes,arrows}

\usepackage{algorithm}
\usepackage{algorithmic}

\newcounter{theorem}
\newcounter{definition}
\newcounter{lemma}
\newcounter{claim}
\newcounter{problem}
\newcounter{proposition}
\newcounter{corollary}
\newcounter{construction}
\newcounter{example}
\newcounter{xca}
\newcounter{comments}
\newcounter{remark}
\newcounter{assumption}

\newtheorem{theorem}[theorem]{Theorem}
\newtheorem{lemma}[lemma]{Lemma}

\newtheorem{corollary}[corollary]{Corollary}

\newtheorem{definition}[definition]{Definition}

\newtheorem{remark}[remark]{Remark}
\newtheorem{assumption}[assumption]{Assumption}

\numberwithin{equation}{section}

\DeclareFontFamily{U}{stix2bb}{}
\DeclareFontShape{U}{stix2bb}{m}{n} {<-> stix2-mathbb}{}

\NewDocumentCommand{\stixbbdigit}{m}{%
	\text{\usefont{U}{stix2bb}{m}{n}#1}%
}
\newcommand{\bbzero}{\stixbbdigit{0}}

\usepackage[many]{tcolorbox}
\usetikzlibrary{calc}
\tcbuselibrary{skins}

\newtcolorbox{resp}[1][]{%
	enhanced jigsaw,%
	colback=gray!5!white,%
	colframe=gray!80!black,%
	size=small,%
	boxrule=1pt,%
	halign title=flush center,%
	coltitle=black,%
	breakable,%
	drop shadow=black!50!white,%
	attach boxed title to top left={xshift=1cm,yshift=-\tcboxedtitleheight/2,yshifttext=-\tcboxedtitleheight/2},%
	minipage boxed title=3cm,%
	boxed title style={%
		colback=white,%
		size=fbox,%
		boxrule=1pt,%
		boxsep=2pt,%
		underlay={%
			\coordinate (dotA) at ($(interior.west) + (-0.5pt,0)$);
			\coordinate (dotB) at ($(interior.east) + (0.5pt,0)$);
			\begin{scope}[gray!80!black]
				\fill (dotA) circle (2pt);
				\fill (dotB) circle (2pt);
			\end{scope}
		}%
	},%
	#1%
}

\usepackage{xspace}

\newcommand{\R}{{\mathbb{R}}}
\newcommand{\Rpz}{{\mathbb{R}_{\geq 0}}}
\newcommand{\Rp}{{\mathbb{R}_{> 0}}}
\newcommand{\N}{{\mathbb{N}}}
\newcommand{\Np}{{\mathbb{N}_{\geq 1}}}
\newcommand{\I}{{\mathbb{I}}}
\newcommand{\One}{{\mathbb{1}}}
\newcommand{\Subsys}{{\textsc{ctia}-UNCS}}
\newcommand{\V}{{\mathbfcal{V}}}

\newcommand{\MPI}{{\mathbfcal{P}_i}}

\newcommand{\Ii}{{\mathbfcal{I}_i}}
\newcommand{\Si}{{\mathbfcal{S}_i}}
\newcommand{\Wi}{{\mathbfcal{W}_i}}
\newcommand{\Sip}{{\mathbfcal{S}^+_i}}
\newcommand{\PD}{{\pmb{\Delta}_i}}

\def\BibTeX{{\rm B\kern-.05em{\sc i\kern-.025em b}\kern-.08em
		T\kern-.1667em\lower.7ex\hbox{E}\kern-.125emX}}
	
	\usepackage{etoolbox}
	
	\makeatletter
	\AtBeginDocument{%
		\pagestyle{headings}
		
		\patchcmd{\@oddhead}{\\[-19pt]}{\\[-8pt]}{}{}%
		\patchcmd{\@evenhead}{\\[-19pt]}{\\[-8pt]}{}{}%
	}
	\makeatother
	
\definecolor{blue(ryb)}{rgb}{0.01, 0.28, 1.0}
\definecolor{fashionfuchsia}{rgb}{0.96, 0.0, 0.63}
\makeatletter
\let\NAT@parse\undefined
\makeatother
\usepackage[colorlinks=true, citecolor=fashionfuchsia, linkcolor=blue(ryb), final]{hyperref}
\renewcommand{\theequation}{\arabic{equation}}
\counterwithout*{equation}{section}

\makeatletter
\def\@opargbegintheorem#1#2#3{\textit{#1\ #2} \textit{(#3):}}
\makeatother

\begin{document}
	
\title{From Data to Sliding Mode Control of Uncertain Large-Scale Networks with Unknown Dynamics}
 \author{Behrad Samari, \IEEEmembership{Student Member,~IEEE}, Gian Paolo Incremona, \IEEEmembership{Senior Member,~IEEE},
 	\\ Antonella Ferrara, \IEEEmembership{Fellow,~IEEE}, and Abolfazl Lavaei, \IEEEmembership{Senior Member,~IEEE}
 	\thanks{B. Samari and A. Lavaei are with the School of Computing, Newcastle University, NE4 5TG Newcastle upon Tyne, United Kingdom (e-mails: {\tt\small{\{b.samari2,abolfazl.lavaei\}@newcastle.ac.uk}}).}
 	 \thanks{G. P. Incremona is with the Dipartimento di Elettronica, Informazione e Bioingegneria, Politecnico di Milano, Piazza Leonardo da Vinci 32, 20133, Milano, Italy (e-mail: {\tt\small{gianpaolo.incremona@polimi.it}}). Supported by the Italian Ministry for Research in the framework of the PRIN 2022 PRIDE, grant no. 2022LP77J4.}
 	\thanks{A. Ferrara is with the Dipartimento di Ingegneria Industriale e dell'Informazione,
         University of Pavia, Via Ferrata 5, 27100 Pavia, Italy (e-mail: {\tt\small{antonella.ferrara@unipv.it}}).}
}

\maketitle
\begin{abstract}
	In this paper, we develop a \emph{compositional data-driven approach} for the global stabilization of large-scale \emph{nonlinear} networks with unknown dynamics and external perturbations. We first collect data along a single trajectory of each unknown nominal subsystem during a finite-time experiment. The data collected from each nominal subsystem are then used to design a feedback law that renders each nominal closed-loop subsystem input-to-state stable (ISS), certified by its corresponding ISS Lyapunov function. We derive conditions as data-dependent semidefinite programs that simultaneously yield local ISS controllers and the corresponding ISS Lyapunov functions. To cancel the effect of external perturbations on subsystem dynamics and, consequently, on the whole network dynamics, we then design a local \emph{integral sliding mode (ISM)} controller for each subsystem using the collected data. Under a \emph{small-gain} compositional condition, we employ data-driven ISS Lyapunov functions designed for the subsystems to construct a \emph{control Lyapunov function} for the network, guaranteeing that the nominal closed-loop network is globally asymptotically stable (GAS) at the origin. We then extend this compositional result to perturbed networks, proving that the synthesized ISM controllers render the origin of the closed-loop network GAS even in the presence of perturbations. We demonstrate the efficacy of the proposed data-driven approach on large-scale interconnected networks with five distinct interconnection topologies.
\end{abstract}

\begin{IEEEkeywords}
	Large-scale networks, data-driven control, integral sliding mode control, compositional techniques.
\end{IEEEkeywords}

\section{Introduction}\label{Sec: Introduction}
\IEEEPARstart{I}{n} recent years, there has been significant interest in developing control frameworks for \emph{large-scale interconnected networks} due to their wide applicability in modeling real-world systems, ranging from power networks to biological ones. However, formally designing control schemes to ensure the stability of such networks remains inherently challenging, primarily due to the mutual dependencies among subsystems. Additionally, modern applications introduce two further challenges that exacerbate these difficulties: \emph{(i)} the lack of an exact mathematical model for each subsystem, and, consequently, for the entire network, and \emph{(ii)} the presence of \emph{external unknown-but-bounded perturbations}.

To address the challenge of mutual dependencies among subsystems in an interconnected network, the literature advocates input-to-state stability (ISS), which examines the network's robust behavior under internal disturbances treated as adversarial inputs~\cite{isidori1985nonlinear, neisic2002integral, agrachev2008input}. By applying \emph{small-gain reasoning}, the overall network stability is guaranteed afterward if the constructed gains adhere to specific compositional conditions~\cite{jiang1994small, jiang1996lyapunov, dashkovskiy2007iss, dashkovskiy2010small, kawan2020lyapunov}. While promising, these approaches commonly assume access to an \emph{exact mathematical model} of the network.

To address challenges arising from perturbations, sliding mode control (SMC) has been widely adopted over the past two decades due to its ability to provide simple, low-computation solutions for various control problems in the presence of uncertainty~\cite{utkin1992smc,incremona2019smc}.
Its design relies on the appropriate selection of a \emph{sliding manifold} to ensure that the system state remains confined to this manifold under a discontinuous control law, thus giving rise to a sliding mode~\cite[Chapter~7]{slotine1991nonlin}. The original SMC formulation, however, has certain limitations: it can induce a chattering effect, does not inherently account for state constraints, and requires a transient phase, known as the \emph{reaching phase}, to reach the manifold, during which robustness to uncertainties is not guaranteed. Among other methods, higher-order SMCs and SMCs with optimal reaching (\emph{e.g.},~\cite{levant2014hosmc,dinuzzo2009hosmc}) provide solutions to chattering alleviation and to constrained control problems, whereas the so-called \emph{integral sliding mode (ISM) control}~\cite{utkin1996integral} enforces robustness from the outset by eliminating the reaching phase.

The ISM control law comprises two components: the so-called \emph{ideal controller}, which stabilizes the nominal dynamics, and a discontinuous component designed to reject uncertainties. Moreover, the sliding variable also consists of two terms: one designed as in classical SMC~\cite[Chapter~1]{incremona2019smc}, and the other, known as the \emph{transient function}, which depends on the nominal dynamics controlled by the ideal controller. This approach has shown effectiveness when integrated with other control strategies to address constraints or optimize performance (see, \emph{e.g.},~\cite{incremona2017mpcism,incremona2017robot}).

To address the challenge of the unavailability of exact mathematical models, two distinct yet valuable data-driven approaches have been proposed: \textit{(i)} indirect data-driven methods and \textit{(ii)} direct data-driven methods. In the \emph{indirect} approach, system identification schemes are used to derive a valid model, which is then leveraged using model-based control techniques. However, accurately identifying the system, particularly when dealing with complex nonlinear dynamics, can be computationally demanding, if not infeasible~\cite{kerschen2006past, hou2013model}. In contrast, direct data-driven approaches aim to bypass the identification phase and directly use data to learn the desired controller. This approach eliminates the need for a two-step process, as required in indirect data-driven methods, enabling a more streamlined solution~\cite{dorfler2022bridging}.

\subsection{Related Work on Data-Based Robust Control}

In the pursuit of stabilizing unknown dynamical systems, multiple data-driven approaches have recently been proposed. To mention a few, inspired by Willems \emph{et al.’s} fundamental lemma~\cite{willems2005note}, \cite{guo2021data} introduces a data-driven framework for stabilizing input-affine nonlinear systems with polynomial dynamics, while~\cite{strasser2021data} extends this approach to more general systems beyond polynomials. More precisely,~\cite{strasser2021data} first considers nonlinear systems with rational dynamics, where both numerators and denominators are polynomials, and later generalizes the method to arbitrary nonlinearities via polynomialization. However, this step increases the state dimension, reducing scalability, and introduces additional conservatism~\cite[Remark~9]{strasser2021data}. In addition, the framework in~\cite{strasser2021data} does not account for nonvanishing disturbances in the system dynamics and therefore does not design ISS controllers; rather, robustness is incorporated to address measurement/process noise. Moreover, whereas~\cite{strasser2021data} formulates feasibility conditions using sum-of-squares (SOS) programs, our framework expresses them directly as semidefinite programs (SDPs), using data-dependent affine matrix equalities and linear matrix inequalities (LMIs) without SOS constraints, which can offer improved scalability and numerical tractability, even at the subsystem level. Finally, unlike~\cite{strasser2021data}, which is limited to monolithic systems, our framework naturally handles large-scale interconnected networks.

The work~\cite{berberich2020data} introduces a robust data-driven model-predictive control strategy for linear time-invariant systems, while~\cite{taylor2021towards} develops a data-driven approach for robust control synthesis of nonlinear systems, explicitly addressing model uncertainty.
In~\cite{de2023learning}, a data-driven scheme is proposed to (approximately) cancel system nonlinearities and achieve stabilization.
A data-driven approach for the certified learning of \emph{incremental} ISS controllers for unknown nonlinear polynomial dynamics has recently been studied in~\cite{zaker2024certified}.
Data-driven stability analyses for unknown systems are further explored in various settings, including switched systems~\cite{kenanian2019data}, continuous-time systems~\cite{boffi2021learning}, and discrete-time systems~\cite{lavaei2022data}. Using noisy measurement data,~\cite{chen2024data} investigates ISS controller design for unknown input-affine nonlinear systems subject to actuator and process disturbances. Its results, however, are limited to systems with polynomial dynamics, and the feasibility conditions are cast as SOS programs, leading to scalability challenges. Furthermore, the results in~\cite{chen2024data} are limited to monolithic systems and cannot generally address large-scale networks due to \emph{computational complexity}.

Recent efforts have sought to extend stability analysis to higher-dimensional systems. In this regard,~\cite{sperl2023approximation} explores the use of neural networks (NNs) to approximate control Lyapunov functions (CLFs), aiming to mitigate the curse of dimensionality. Leveraging the compositional structure of interconnected nonlinear systems,~\cite{liu2024compositionally} introduces a framework for training and validating neural Lyapunov functions for stability analysis. Similarly,~\cite{zhang2023compositional} employs a compositional approach to derive neural certificates based on ISS Lyapunov functions and controllers for networked dynamical systems. Despite these advancements,~\cite{sperl2023approximation, liu2024compositionally, zhang2023compositional} assume that the system model is available. Additionally,~\cite{zhang2023compositional} would require a second formal-verification step, \emph{e.g.}, one involving Satisfiability Modulo Theories (SMT) solvers to verify the correctness of the candidate ISS Lyapunov functions. However, SMT solvers may encounter termination issues or face scalability challenges depending on the NN size or the system complexity~\cite{wongpiromsarn2015automata,abate2022neural}. Then, even if the problem is solved for \emph{high-dimensional} systems using NNs, verifying the results in the second step would potentially become intractable using SMT solvers, which may lead to a lack of formal guarantees.

Finally,~\cite{lavaei2023data} proposes a \emph{scenario-based} method~\cite{calafiore2006scenario} for verifying the stability of interconnected networks. However, this approach requires \emph{independent and identically distributed (i.i.d.)} data, meaning only one input--state sample pair can be obtained per trajectory~\cite{esfahani2014performance}. As a result, multiple independent trajectories are needed to certify network stability. In contrast, our approach requires data collected along a single trajectory of each nominal subsystem. Moreover, while the method in~\cite{lavaei2023data} is restricted to nonlinear \emph{homogeneous} systems and focuses solely on \emph{verification}, our framework accommodates a much broader class of nonlinear systems and enables \emph{controller synthesis} in the presence of disturbances.
\subsection{Related Work on Data-Based Sliding Mode Control}
Recently, limited studies have explored the problem of designing sliding mode controllers from data. In this direction, in~\cite{riva2024ddsmc}, a data-driven framework is proposed that integrates virtual reference feedback tuning (VRFT)~\cite{campi2022vrft} with ISM control for \emph{single-input single-output} perturbed linear systems. To extend beyond linear systems,~\cite{lan2025data} introduces a data-driven framework for designing sliding mode controllers for discrete-time nonlinear systems, assuming the control matrix is known. In contrast, our framework is developed for continuous-time systems and requires no prior knowledge of the control matrix. Moreover,~\cite{lan2025data} presents a data-driven counterpart of classical SMC that includes a reaching phase, during which exact disturbance rejection is not ensured, and guarantees convergence only to a neighborhood of the manifold, the size of which depends on the disturbance and tuning parameters. In contrast, our framework eliminates the reaching phase and achieves exact rejection of matched perturbations from the outset through the data-driven design of ISM controllers, guaranteeing that the closed-loop network is globally asymptotically stable (GAS) at the origin. We note that neither~\cite{riva2024ddsmc} nor~\cite{lan2025data} investigates \emph{large-scale interconnected networks} consisting of multiple coupled subsystems with unknown dynamics.

\subsection{Main Contribution}
Motivated by the research gap mentioned above, this paper proposes a compositional data-driven methodology to achieve global stabilization of unknown large-scale nonlinear networks subject to external disturbances. The design process involves deriving conditions formulated as data-dependent SDPs, whose feasibility enables the construction of local ISS controllers and the associated ISS Lyapunov functions. To reject the influence of external perturbations on the dynamics of individual subsystems and, consequently, the entire interconnected network, a local ISM controller is developed for each subsystem based on the collected data. By leveraging a \emph{small-gain} compositional condition, the data-driven ISS Lyapunov functions designed for each nominal subsystem are combined to construct a CLF for the nominal interconnected network, ensuring that the nominal closed-loop network is GAS at the origin. Our compositional approach is further extended to networks subject to external perturbations, demonstrating that the ISM controllers render the origin of the closed-loop network GAS even under external perturbations.

\section{Problem Formulation}\label{Sec: Prob_Desc}

\subsection{Notation}
The symbol $\R$ represents the set of real numbers, while $\Rpz$ and $\Rp$ denote the sets of non-negative and positive real numbers, respectively. The set of non-negative integers is given by $\N$, while $\Np$ represents the set of positive integers. The $n \times n$ identity matrix is represented by $\I_n$, and the $n$-dimensional vector of ones is denoted by $\One_n$. The notation $\bbzero_{n \times m}$ denotes the $n \times m$ zero matrix, while the zero vector in $\mathbb{R}^n$ is denoted by $\bbzero_n$. For $N$ vectors $x_i \in \R^{n_i}$, $x = [x_1;\ldots;x_N]$ represents the corresponding column vector of dimension $\sum_{i=1}^N n_i$. Additionally, $[x_1\;\ldots\; x_N]$ and $[A_1\;\ldots\; A_N]$ denote the horizontal concatenation of vectors $x_i \in \R^n$ and matrices $A_i \in \R^{n \times m}$, resulting in matrices of dimensions $n \times N$ and $n \times mN$, respectively. The Euclidean norm of a vector $x \in \R^n$ is denoted by $\vert x \vert$, and the induced 2-norm of a matrix $A$ is represented by $\Vert A \Vert$. Given a symmetric matrix $P$, $P \succ 0$ ($P \succeq 0$) implies that $P$ is positive (semi)definite. For a symmetric matrix $P$, its minimum and maximum eigenvalues are denoted by $\lambda_{\min}(P)$ and $\lambda_{\max}(P)$, respectively. The transpose of a matrix $B$ is written as $B^\top \!$. A block diagonal matrix with diagonal blocks $(A_1, \ldots, A_N)$ is denoted by $\mathsf{blkdiag}(A_1, \ldots, A_N)$, while a diagonal matrix with \emph{scalar} entries $(a_1, \ldots, a_N)$ is represented by $\mathsf{diag}(a_1, \ldots, a_N)$. The notation $\{a_{ij}\}$ denotes the construction of a matrix whose elements $a_{ij}$ are located in the $i$-th row and $j$-th column. A function $\beta: \Rpz \rightarrow \Rpz$ belongs to class $\mathcal{K}$ if it is continuous, strictly increasing, and satisfies $\beta(0) = 0$. Similarly, a function $\beta: \Rpz \times \Rpz \rightarrow \Rpz$ belongs to class $\mathcal{KL}$ if, for a fixed $s$, the mapping $\beta(r,s)$ belongs to class $\mathcal{K}$ with respect to $r$, and for a fixed $r \in \Rp$, the mapping $\beta(r,s)$ is decreasing with respect to $s$ and approaches $0$ as $s \rightarrow \infty$. For a given matrix \(A\), \(\mathbf{b} \in \operatorname{col}(A)\) means that the vector \(\mathbf{b}\) can be expressed as a linear combination of the columns of \(A\), \emph{i.e.}, \(\mathbf{b}\) lies in the column space of \(A\). Throughout the paper, time arguments are omitted whenever no ambiguity arises.

\subsection{Individual Subsystems}
We consider subsystems with continuous-time input-affine uncertain nonlinear dynamics, as formalized in the following definition. These subsystems are subsequently interconnected to construct the large-scale network.

\begin{definition}[\textbf{\Subsys}]\label{def: ctia-UNCS}
	A continuous-time input-affine uncertain nonlinear control system (\Subsys) is described by
	\begin{align}
		\Upsilon_i : \dot{x}_i = f_i(x_i) + \mathcal{B}_i u_i + \mathcal{D}_i w_i + \mathcal{E}_i (x_i, t), \label{eq: original_subsys}
	\end{align}
	where $x_i \in \R^{n_i}$ is the state, $f_i: \R^{n_i} \rightarrow \R^{n_i}$ is a continuously differentiable vector field, with $f_i(\bbzero_{n_i}) = \bbzero_{n_i}$, $\mathcal{B}_i \in \R^{n_i \times m_i}$, with $\operatorname{rank}(\mathcal{B}_i)=m_i$, denotes the control matrix, $u_i \in \R^{m_i}$ is the \emph{control} input, and $\mathcal{E}_i : \R^{n_i} \times \Rpz \rightarrow \R^{n_i}$ denotes the \emph{external perturbation}, satisfying the  \emph{matching} condition, \emph{i.e.}, $\mathcal{E}_i(x_i, t) \in \operatorname{col}(\mathcal{B}_i)$ for all $(x_i, t) \in \R^{n_i} \times \Rpz$.
	Moreover, $\mathcal{D}_i \in \R^{n_i \times \psi_i}$ is the \emph{internal input} matrix, where $\psi_i = \sum_{j = 1,\, j\neq i}^{N} n_j$ and $N$ is the number of subsystems. The matrix $\mathcal{D}_i$ is partitioned as
	\begin{subequations}\label{eq: partitioned}
		\begin{align}
			\mathcal{D}_i  & = [\mathcal{D}_{i 1} \;\; \ldots \;\; \mathcal{D}_{i(i-1)} \;\; \mathcal{D}_{i(i+1)} \;\; \ldots \;\; \mathcal{D}_{i N}], \label{eq: partitioned_D}
		\end{align}
		where $\mathcal{D}_{ij}\in\R^{n_i\times n_j}$. Additionally, $w_i \in \R^{\psi_i}$ represents the \emph{internal} input, \emph{i.e.}, the input capturing the effect of other subsystems within the network, partitioned as
		\begin{align}
			w_i  &= [w_{i 1} ; \ldots ; w_{i(i-1)} ; w_{i(i+1)} ; \ldots ; w_{i N}], \label{eq: partitioned_w}
		\end{align}
		where $w_{ij}=x_j$ if subsystem $j$ influences subsystem $i$, and $w_{ij}=\bbzero_{n_j}$ otherwise, according to the interconnection topology (cf.~\eqref{eq: internal-connection}).
	\end{subequations}
	Finally, it is assumed that, for every initial condition \(x_i(0)\in\R^{n_i}\), system~\eqref{eq: original_subsys} admits a forward-complete solution defined on \([0,\infty)\) under the considered control input, internal input, and external perturbation signals\footnote{For continuous control inputs, solutions are assumed to be unique and understood in the classical sense, whereas, when the discontinuous ISM input is introduced later, they are interpreted in the standard sliding-mode/Filippov sense.}.
\end{definition}

Without loss of generality, the dynamics of the \Subsys\ in~\eqref{eq: original_subsys} can be parameterized as
\begin{align}
		\Upsilon_i : \dot{x}_i = \mathcal{A}^{\prime}_i \mathcal{Z}^{\prime}_i(x_i) + \mathcal{B}_i u_i + \mathcal{D}_i w_i + \mathcal{E}_i (x_i, t), \label{eq: final_subsys_prime}
\end{align}
where $\mathcal{Z}^{\prime}_i : \R^{n_i} \rightarrow \R^{z^{\prime}_i}$, with $\mathcal{Z}^{\prime}_i(\bbzero_{n_i}) = \bbzero_{z^{\prime}_i}$, is a continuously differentiable vector function containing the exact basis functions of $f_i$, and $\mathcal{A}^{\prime}_i \in \R^{n_i \times z^{\prime}_i}$ is the system matrix.

\begin{remark}[\textbf{Transition From~\eqref{eq: original_subsys} to~\eqref{eq: final_subsys_prime}}]\label{rem:ref}
		Since $\mathcal{Z}^{\prime}_i$ includes the exact basis functions of $f_i$, each component of $f_i$ can be represented as a linear combination of the entries of $\mathcal{Z}^{\prime}_i$. Hence, for each $k \in \{1,\dots,n_i\}$, there exists $a_k^i \in \R^{z^{\prime}_i}$ such that
		\[
		f_{i_k}(x_i) = (a_k^i)^\top \mathcal{Z}^{\prime}_i(x_i), \quad \forall x_i \in \R^{n_i},
		\]
		where $f_{i_k}$ denotes the $k$-th component of $f_i$.
		Accordingly, $f_i(x_i)$ can be expressed as $\mathcal{A}^{\prime}_i \mathcal{Z}^{\prime}_i(x_i)$ for all $x_i \in \R^{n_i}$, where $\mathcal{A}^{\prime}_i$ is obtained by vertically stacking $(a_k^i)^\top$ for all $k$, showing that the transition from~\eqref{eq: original_subsys} to~\eqref{eq: final_subsys_prime} is without loss of generality.
\end{remark}

In this paper, $\mathcal{A}^{\prime}_i$ and $\mathcal{B}_i$, along with the vector functions $\mathcal{Z}^{\prime}_i$ and $\mathcal{E}_i$, are assumed to be \emph{unknown}. However, we make the following assumption regarding $\mathcal{Z}^{\prime}_i$.

\begin{assumption}[\textbf{On $\mathcal{Z}^{\prime}_i$}~\cite{monshizadeh2025versatile}]\label{assumpt: dictionary}
		A continuously differentiable vector function $\mathcal{Z}_i$, referred to as a \emph{dictionary}, is known, satisfying
		\begin{align*}
			\mathcal{Z}^{\prime}_i(x_i) = \digamma_{\! i} \mathcal{Z}_i(x_i), \quad \forall x_i \in \R^{n_i},
		\end{align*}
		for some unknown matrix $\digamma_{\! i} \in \R^{z^{\prime}_i \times z_i}$, where $z_i \geq z^{\prime}_i$.
\end{assumption}
Under Assumption~\ref{assumpt: dictionary}, the \Subsys\ in~\eqref{eq: original_subsys}, and accordingly in~\eqref{eq: final_subsys_prime}, can be rewritten as
\begin{align}
	\Upsilon_i : \dot{x}_i = \mathcal{A}_i \mathcal{Z}_i(x_i) + \mathcal{B}_i u_i + \mathcal{D}_i w_i + \mathcal{E}_i (x_i, t), \label{eq: final_subsys}
\end{align}
where $\mathcal{A}_i \in \R^{n_i \times z_i}$, defined as $\mathcal{A}_i \coloneq \mathcal{A}^{\prime}_i \digamma_{\! i}$, is unknown, and $\mathcal{Z}_i : \R^{n_i} \rightarrow \R^{z_i}$ is the known dictionary satisfying $\mathcal{Z}_i(\bbzero_{n_i}) = \bbzero_{z_i}$.
This dictionary consists of both linear and nonlinear components, expressed as
\begin{align}
	\mathcal{Z}_i(x_i) = \begin{bmatrix}
		x_i\vspace{-.225cm}\\
		\tikz\draw [thin,dashed] (0,0) -- (1.25,0);\\
		\mathcal{M}_i(x_i)
	\end{bmatrix}\!\!, \label{eq: dictionary}
\end{align}
where $\mathcal{M}_i : \R^{n_i} \rightarrow \R^{z_i - n_i}$ contains only \emph{nonlinear} basis functions. In this work, we assume that the nonlinear contribution associated with $\mathcal{M}_i$ is \emph{matched}, meaning that these nonlinear terms appear in the same channel as the control input $u_i$, thereby allowing their direct cancellation by $u_i$. We use the tuple $\Upsilon_i = (\mathcal{A}_i, \mathcal{Z}_i, \mathcal{B}_i, \mathcal{D}_i, \R^{n_i}, \R^{m_i}, \R^{\psi_i}, \R^{n_i})$ to represent the reformulated \Subsys\ in~\eqref{eq: final_subsys}. We note that the proposed framework is not robust to mismatches in the dictionary $\mathcal{Z}_i$; that is, if the dictionary is not sufficiently rich to represent all basis functions of $f_i$ (or $\mathcal{Z}^{\prime}_i$), the framework becomes inapplicable.

\begin{remark}[\textbf{Matched Nonlinearities}]\label{Remark: matched_non}
		Note that the nonlinear contribution induced by $\mathcal{M}_i$ satisfies the same matching condition as the external perturbation $\mathcal{E}_i$.
		To elaborate, let us define $\mathcal{A}_i \coloneq \begin{bmatrix}
			\mathcal{A}_{i_1} & \mathcal{A}_{i_2}
		\end{bmatrix}$, where $\mathcal{A}_{i_1} \in \R^{n_i \times n_i}$ and $\mathcal{A}_{i_2} \in \R^{n_i \times (z_i - n_i)}$. One can then rewrite~\eqref{eq: final_subsys} as
		\begin{align*}
			\Upsilon_i : \dot{x}_i &= \mathcal{A}_{i_1} x_i + \mathcal{A}_{i_2} \mathcal{M}_i(x_i) + \mathcal{B}_i u_i + \mathcal{D}_i w_i + \mathcal{E}_i (x_i, t).
		\end{align*}
		Hence, the matching condition for the system nonlinearities is 
		\begin{align*}
		\mathcal{A}_{i_2} \mathcal{M}_i(x_i) \in \operatorname{col}(\mathcal{B}_i), \quad \forall x_i \in \R^{n_i},
		\end{align*}
		which implies the existence of a continuously differentiable vector function $\aleph_i : \R^{n_i} \to \R^{m_i}$ such that
		\begin{align*}
			\mathcal{A}_{i_2} \mathcal{M}_i(x_i) = \mathcal{B}_i \aleph_i(x_i), \quad \forall x_i \in \R^{n_i},
		\end{align*}
		indicating that the nonlinear contribution can be directly canceled by the control input.
\end{remark}

While the matrices $\mathcal{A}_i$ and $\mathcal{B}_i$, together with the external perturbation $\mathcal{E}_i$, are assumed to be \emph{unknown}, the matrix $\mathcal{D}_i$ is assumed to be known, as it encapsulates the \emph{interconnection weights} between subsystems, which are often known a priori in interconnected network settings. Thus, the \Subsys\ $\Upsilon_i$ is referred to as a system with a \emph{(partially)} unknown model, as its matrices $\mathcal{A}_i$ and $\mathcal{B}_i$ are entirely unknown, while the dictionary $\mathcal{Z}_i$ is available.

\begin{remark}[\textbf{Dictionary $\mathcal{Z}_i$}]\label{Remark: dictionary}
	 Assumption~\ref{assumpt: dictionary} essentially implies that the dictionary $\mathcal{Z}_i$ is extensive enough to capture all possible terms in the actual dynamics, albeit at the cost of incorporating \emph{extraneous terms}.
	 This assumption is not overly restrictive in many cases. This is because, in numerous practical applications, such as electrical and mechanical engineering systems, system dynamics can often be deduced from first principles, which naturally align with the basis functions included in~\eqref{eq: dictionary}. However, the exact system parameters frequently remain unknown. This is consistent with our assumption that the matrices $\mathcal{A}_i$ and $\mathcal{B}_i$ are entirely unknown.
\end{remark}

We assume that the matched perturbation $\mathcal{E}_i$ is bounded, as commonly assumed in SMC theory.

\begin{assumption}[\textbf{Bound on $\mathcal{E}_i$ \cite[Chapter~1]{incremona2019smc}}]\label{Assumption: bound_disturbance}
	There exists a known constant $\bar{\mathbfcal{E}}_i \in \Rp$ satisfying $\vert \mathcal{E}_i(x_i, t) \vert \leq \bar{\mathbfcal{E}}_i$, for all $(x_i, t) \in \R^{n_i} \times \Rpz$.
\end{assumption}

To conclude this subsection, we finally note that hereafter the \emph{nominal} model of the \Subsys\ $\Upsilon_i$, \emph{i.e.,} the \Subsys\ $\Upsilon_i$ in~\eqref{eq: final_subsys} when $\mathcal{E}_i(x_i, t) \equiv \bbzero_{n_i}$ for all \((x_i,t)\in\R^{n_i}\times\Rpz\), is denoted by $\Upsilon_i^\ast$.

 \begin{remark}[\textbf{On Definition~\ref{def: ctia-UNCS}}]\label{Remark: on_def1}
		It is evident that the \Subsys\ in~\eqref{eq: original_subsys}, and accordingly in~\eqref{eq: final_subsys}, can equivalently be expressed as
		\begin{align*}
			\Upsilon_i : \dot{x}_i = \mathcal{A}_i \mathcal{Z}_i(x_i) + \mathcal{B}_i u_i + \sum_{j\in N_i} \mathcal{D}_{ij} x_j + \mathcal{E}_i (x_i, t),
		\end{align*}
		where $N_i$ denotes the set of neighbors of the $i$-th subsystem.
\end{remark}

\subsection{Interconnected Network}
Here, we formally define the interconnection among the subsystems $\Upsilon_i$. Let $\Upsilon$ represent an interconnected network comprising $N \in \Np$ subsystems $\Upsilon_i$, with internal inputs and their associated matrices  as in~\eqref{eq: partitioned}.
If no connection exists from subsystem $\Upsilon_j$ to $\Upsilon_i$, the corresponding effect on the $i$-th subsystem is identically zero, \emph{i.e.}, $\mathcal{D}_{ij} w_{ij} \equiv \bbzero_{n_i}$ ($\mathcal{D}_{ij} = \bbzero_{n_i \times n_j}$ and $w_{ij} = \bbzero_{n_j}$). We now move forward with the formal definition of an interconnected network.

\begin{definition}[\textbf{Interconnected Network}]\label{def: network}
	Consider $N \in \Np$ subsystems $\Upsilon_i = (\mathcal{A}_i, \mathcal{Z}_i, \mathcal{B}_i, \mathcal{D}_i, \R^{n_i}, \R^{m_i}, \R^{\psi_i}, \R^{n_i})$, $i \in \{1, \ldots, N\}$, with internal inputs and their associated matrices partitioned as described in~\eqref{eq: partitioned}. Suppose that the subsystems $\Upsilon_i$, $i \in \{1, \ldots, N\}$, are interconnected under the following interconnection constraint:
	\begin{align}
		w_{ij} = \begin{cases}
				x_j, & j \in N_i,\\
				\bbzero_{n_j}, & j \notin N_i,\end{cases} \label{eq: internal-connection}
	\end{align}
	where $N_i$ denotes the neighbor set of the $i$-th subsystem.
	Then an interconnected network $\Upsilon = \mathscr{N}(\Upsilon_1, \ldots, \Upsilon_N)$ is described by
	\begin{align}
		\Upsilon : \dot{x} = \mathcal{A}\mathcal{Z}(x) + \mathcal{B}u + \mathcal{E}(x, t), \label{eq: network}
	\end{align}
	where $\mathcal{A} \in \R^{n \times z},$ in which $n \coloneq \sum_{i=1}^{N} n_i$ and $z \coloneq \sum_{i=1}^{N} z_i,$ is a block matrix with diagonal blocks $(\mathcal{A}_1, \ldots, \mathcal{A}_N)$ and topology-dependent off-diagonal blocks in its $i$-th block row, horizontally concatenated as $\hat{\mathcal{D}}_i = [\hat{\mathcal{D}}_{i 1} \; \ldots \; \hat{\mathcal{D}}_{i(i-1)} \; \hat{\mathcal{D}}_{i(i+1)} \; \ldots \; \hat{\mathcal{D}}_{i N}]$, where $\hat{\mathcal{D}}_{ij}  = [\mathcal{D}_{ij} \;\; \bbzero_{n_i \times (z_j - n_j)}]$. Moreover, $\mathcal{E}(x, t) = [\mathcal{E}_1(x_1, t); \ldots; \mathcal{E}_N(x_N, t)]$ for all $(x, t) \in \R^{n} \times \Rpz$, and $\mathcal{B} = \mathsf{blkdiag}(\mathcal{B}_1, \ldots, \mathcal{B}_N) \in \R^{n \times m}$, with $u \in \R^{m}$ and $m \coloneq \sum_{i=1}^N m_i$, while $\mathcal{Z}(x) = [\mathcal{Z}_1(x_1); \dots; \mathcal{Z}_N(x_N)] \in \R^{z}$ for all $x \in \R^{n}$. We use the tuple $\Upsilon = (\mathcal{A}, \mathcal{Z}, \mathcal{B}, \R^n, \R^m, \R^n)$ to represent the interconnected network in \eqref{eq: network}.
\end{definition}

Following the approach used to define the nominal model of each subsystem, henceforth, the \emph{nominal} model of the interconnected network, \emph{i.e.}, the network in~\eqref{eq: network} when $\mathcal{E}(x, t) \equiv \bbzero_{n}$ for all \((x,t)\in\R^{n}\times\Rpz\), is indicated by $\Upsilon^\ast$.

Having defined the interconnected network in Definition~\ref{def: network}, we now present sufficient conditions ensuring GAS of the interconnected network.

\subsection{GAS of Interconnected Networks}
We start by providing a formal definition of the GAS\footnote{Throughout the paper, we use ``GAS'' interchangeably for both global asymptotic stability and globally asymptotically stable; the intended meaning becomes clear in the context.} of an interconnected network, as introduced in~\cite{angeli2000characterization}.

\begin{definition}[\textbf{GAS of Network}]\label{def: GAS}
	The origin is said to be a GAS equilibrium point for an interconnected network $\Upsilon = \mathscr{N}(\Upsilon_1, \ldots, \Upsilon_N)$ if there exists a function $\beta$, belonging to class $\mathcal{KL}$, such that for any initial condition $x(0) \in \R^n$, the inequality
	\begin{align*}
		\vert x(t) \vert \leq \beta(\vert x(0) \vert, \, t), \quad \forall t \geq 0,
	\end{align*}
	holds, implying that all state trajectories of $\Upsilon$ approach the origin, which serves as the equilibrium point, as time $t \rightarrow \infty$.
\end{definition}

Before proceeding, we emphasize that the upcoming theorem and definition are specifically formulated for the \emph{nominal} interconnected network $\Upsilon^\ast$ and \emph{nominal} subsystems $\Upsilon_i^\ast$, $i \in \{1, \ldots, N\}$, respectively. This focus aligns with our objective of designing an ISM controller to mitigate the effects of the perturbations $\mathcal{E}_i$, and accordingly, $\mathcal{E}$.  More precisely, for each subsystem $\Upsilon_i$, the control input is expressed as $u_i = u_{i}^\ast + u_{i}^{\mathrm{ISM}}$, where $u_{i}^\ast \in \R^{m_i}$ represents a state-feedback control law that renders each \emph{nominal} closed-loop subsystem ISS. At the same time, $u_{i}^{\mathrm{ISM}} \in \R^{m_i}$ serves as a \emph{discontinuous} control input intended to make the system robust against perturbations $\mathcal{E}_i(x_i, t)$ for all \((x_i,t)\in\R^{n_i}\times\Rpz\) (cf. Section~\ref{subsec: ISM}).

We now state the following theorem, borrowed from~\cite{sontag1996new}, which outlines the sufficient conditions for the nominal network $\Upsilon^\ast$ to achieve GAS, as per Definition~\ref{def: GAS}.

\begin{theorem}[\textbf{CLFs for Interconnected Networks}]\label{Theorem: model-based Network}
	Consider a nominal interconnected network $\Upsilon^\ast = \mathscr{N}(\Upsilon_1^\ast, \ldots, \Upsilon_N^\ast)$, comprising $N$ nominal subsystems $\Upsilon_i^\ast$, $i \in \{1, \ldots, N\}$. Assume the existence of a {smooth} CLF $\V : \R^n \to \Rpz$ and constants $\alpha_1, \alpha_2, \kappa \in \Rp$ such that
	\begin{itemize}
		\item For all $x \in \R^n$:
		\begin{subequations}
			\begin{equation}\label{eq: clf-con1-network}
				\alpha_1 \vert x \vert^2 \leq \V(x) \leq \alpha_2 \vert x \vert^2,
			\end{equation}
			\item For all $x \in \R^n$, there exists \(u^\ast\in\R^m\) such that\footnote{Here, the existential condition is understood as requiring a locally Lipschitz state-feedback law \(u^\ast=k(x)\), with \(k(\bbzero_n)=\bbzero_m\), satisfying~\eqref{eq: clf-con2-network} at every state; the same interpretation, with a local state-feedback law \(u_i^\ast=k_i(x_i)\), applies in Definition~\ref{def: ISS Lyapunov}.}:
			\begin{align}\label{eq: clf-con2-network}
				&\mathsf{L} \V(x) \leq - \kappa \V(x),
			\end{align}
		\end{subequations}
	\end{itemize}
	where $\mathsf{L} \V$ denotes the \emph{Lie derivative} of $\V : \R^n \to \Rpz$ along the trajectories of the \emph{nominal} dynamics of~\eqref{eq: network}, given by
	\begin{align}
		\mathsf{L} \V(x) = \partial_{x} \V(x)(\mathcal{A} \mathcal{Z}(x) + \mathcal{B} u^\ast),
	\end{align}
	with $\partial_{x} \V(x) = \frac{\partial \V(x)}{\partial x}$. Under these conditions, the \emph{nominal} closed-loop interconnected network $\Upsilon^\ast = \mathscr{N}(\Upsilon_1^\ast, \ldots, \Upsilon_N^\ast)$ is GAS in the sense of Definition~\ref{def: GAS}.
\end{theorem}

Designing a CLF to ensure GAS of the nominal network $\Upsilon^\ast$ is, in most cases, a computationally intensive task, even when the accurate mathematical model is available. To overcome this difficulty, one can instead pursue GAS of the nominal interconnected network $\Upsilon^\ast$ by leveraging the ISS characteristics of its individual nominal subsystems $\Upsilon_i^\ast$, $i \in \{1, \ldots, N\}$, as outlined in the subsequent definition.

\begin{definition}[\textbf{ISS Lyapunov Functions}~ {\cite{sontag1995characterizations}}]\label{def: ISS Lyapunov}
	For a given nominal subsystem $\Upsilon_i^\ast$, a smooth function $\V_i : \R^{n_i} \to \Rpz$ is said to be an ISS Lyapunov function if there exist constants $\alpha_{1_i}, \alpha_{2_i}, \kappa_i \in \Rp$ and $\rho_i \in \Rpz$, fulfilling
	\begin{itemize}
		\item For all $x_i \in \R^{n_i}$:
		\begin{subequations}
			\begin{eqnarray}\label{eq: ISS-con1}
				\alpha_{1_i} \vert x_i \vert^2 \leq \V_i(x_i) \leq \alpha_{2_i} \vert x_i \vert^2,
			\end{eqnarray}
			\item For all $x_i \in \R^{n_i}$, there exists $u_i^\ast \in \R^{m_i}$ such that for all $w_i \in \R^{\psi_i}$:
			\begin{align}\label{eq: ISS-con2}
				\mathsf{L} \V_i(x_i) \leq - \kappa_i \V_i(x_i) + \rho_i \vert w_i \vert^2,
			\end{align}
		\end{subequations}
	\end{itemize}
	where $\mathsf{L} \V_i$ represents the \emph{Lie derivative} of $\V_i : \R^{n_i} \to \Rpz$ along the trajectories of the nominal dynamics of~\eqref{eq: final_subsys}, given by
	\begin{align}
		\mathsf{L} \V_i(x_i) = \partial_{x_i} \V_i(x_i) (\mathcal{A}_i \mathcal{Z}_i(x_i) + \mathcal{B}_i u_i^\ast + \mathcal{D}_i w_i). \label{eq: Lie derivative-subsystem}
	\end{align}
\end{definition}

\subsection{ISM Control}\label{subsec: ISM}
As briefly discussed in the preceding subsection, we aim to design a local controller for each subsystem, represented by $u_i = u_{i}^\ast + u_{i}^{\mathrm{ISM}}$, where $u_{i}^\ast$ is the ideal controller~\cite{utkin1996integral}, and $u_{i}^{\mathrm{ISM}}$ is designed to ensure robustness against matched perturbations $\mathcal{E}_i(x_i, t)$ for all \((x_i,t)\in\R^{n_i}\times\Rpz\) from the initial time instant $t_0 \in \Rpz$\footnote{For notational generality, the ISM control mechanism is initialized at an arbitrary time \(t_0\in\Rpz\); however, in Section~\ref{subsec: ISM_data} and the subsequent analysis, we consider \(t_0=0\) without loss of generality.}. For the case where $m_i > 1$, meaning each subsystem $\Upsilon_i$ has multiple control inputs, the discontinuous control law is derived using the unit vector approach, \emph{i.e.}\footnote{At $\sigma_i=\bbzero_{m_i}$, the term $\frac{\sigma_i}{|\sigma_i|}$ in~\eqref{eq: ISM_U} is interpreted in the standard sliding-mode/Filippov sense.},
\begin{align}
	u_{i}^{\mathrm{ISM}} (t) = -\Theta_i \frac{\sigma_i (x_i(t), t)}{\vert \sigma_i (x_i(t), t) \vert}, \quad \forall t \geq t_0 \geq 0, \label{eq: ISM_U}
\end{align}
where $\Theta_i \in \Rp$ is a constant gain selected so that the worst realization of the perturbation is dominated, while $\sigma_i (x_i(t), t) \in \R^{m_i}$ for all $(x_i, t) \in \R^{n_i} \times \Rpz$ is the integral sliding variable, defined, for all $t \geq t_0$, as
\begin{align}
	\sigma_i(x_i(t), t) \!=\! \sigma_{0_i}(x_i (t)) \!+\! \zeta_i(t), \;\; \sigma_{0_i}(x_i(t_0)) \!=\! -\zeta_i(t_0). \label{eq: integral sliding variable}
\end{align}
The term $\sigma_{0_i} : \R^{n_i} \to \R^{m_i}$ is selected by the designer, possibly as a linear combination of state variables $x_i$. On the other hand, $\zeta_i \in \R^{m_i}$ in~\eqref{eq: integral sliding variable} represents what is known as the \emph{transient function}, characterized by its dynamics defined as
\begin{align}
	\dot{\zeta}_i = -\frac{\partial \sigma_{0_i}(x_i)}{\partial x_i} (\mathcal{A}_i \mathcal{Z}_i(x_i) + \mathcal{B}_i u_i^\ast + \mathcal{D}_i w_i), \label{eq: transient}
\end{align}
where $\frac{\partial \sigma_{0_i}(x_i)}{\partial x_i} \in \R^{m_i \times n_i}$ and $\zeta_i(t_0) = - \sigma_{0_i}(x_i(t_0))$. As demonstrated in~\cite{utkin1996integral}, under the linear choice $\sigma_{0_i}(x_i) = \mathcal{C}_i x_i$, where $\mathcal{C}_i \in \R^{m_i \times n_i}$ is a constant matrix satisfying $\mathcal{C}_i \mathcal{B}_i \succ0$, selecting the constant gain $\Theta_i$ to sufficiently exceed the maximum possible impact of the matched perturbation (cf. Assumption~\ref{Assumption: bound_disturbance}) guarantees that a sliding mode $\sigma_i = \bbzero_{m_i}$ is achieved for all $t \geq t_0$, thereby ensuring that the matched perturbations are consistently rejected (cf. Section~\ref{subsec: ISM_data}).

It is evident that designing ISS Lyapunov functions and transient functions is not feasible due to the presence of unknown matrices $\mathcal{A}_i$ and $\mathcal{B}_i$ in~\eqref{eq: Lie derivative-subsystem} and~\eqref{eq: transient}, necessitating the development of a data-driven framework. To address this challenge, the next section presents the proposed data-driven framework for designing ISS Lyapunov functions and the corresponding controllers for nominal subsystems and ISM controllers for non-nominal ones, followed by Section~\ref{Section: Compositional}, which extends this framework compositionally to large-scale interconnected networks.

\section{Data-Driven Framework}\label{Section: Data Framework}
This section proposes our data-driven framework for synthesizing a local controller $u_i = u_{i}^\ast + u_{i}^{\mathrm{ISM}}$ for each \Subsys\ $\Upsilon_i$, $i \in \{1, \ldots, N\}$. More precisely, Section~\ref{subsec: ISS_data} introduces our approach for designing ISS controllers $u_i^\ast$ and the associated ISS Lyapunov functions $\V_i$, while Section~\ref{subsec: ISM_data} offers the methodology for designing the ISM control, \emph{i.e.}, $u_i^{\mathrm{ISM}}$, for each \Subsys\ $\Upsilon_i$, using data.

\subsection{Data-Driven Design of ISS Lyapunov Functions}\label{subsec: ISS_data}
Within our data-driven framework, the ISS Lyapunov function for each subsystem is assumed to have a \emph{quadratic} structure of the form $\V_i(x_i) = x_i^\top \MPI x_i$, where $\MPI \succ 0$. Data are then collected from each unknown nominal subsystem $\Upsilon_i^\ast$ over the time interval $[t_0, t_0 + (\mathcal{T} - 1)\tau]$, where $\mathcal{T} \in \Np$ represents the number of collected samples, and $\tau \in \Rp$ denotes the sampling time:
\begin{subequations}\label{eq: data_ct} 
	\begin{align}
		\Ii &= \begin{bmatrix}
			u_i(t_0) & \! u_i(t_0 \!+\! \tau) & \! \dots & \! u_i(t_0 \!+\! (\mathcal{T} \!-\! 1)\tau)
		\end{bmatrix}\!\!,\label{eq: Ii}\\
		\Si &= \begin{bmatrix}
			x_i (t_0) & \! x_i (t_0 \!+\! \tau) & \! \dots & \! x_i (t_0 \!+\! (\mathcal{T} \!-\! 1)\tau)
		\end{bmatrix}\!\!,\label{eq: Si}\\
		\Wi &= \begin{bmatrix}
			w_i (t_0) & \! w_i (t_0 \!+\! \tau) & \! \dots & \! w_i (t_0 \!+\! (\mathcal{T} \!-\! 1)\tau)
		\end{bmatrix}\!\!,\label{eq: Wi}\\
		\Sip &= \begin{bmatrix}
			\dot{x}_i (t_0) & \! \dot{x}_i (t_0 \!+\! \tau) & \! \dots & \! \dot{x}_i (t_0 \!+\! (\mathcal{T} \!-\! 1)\tau)
		\end{bmatrix}\!\!.\label{eq: Sip}
	\end{align}
\end{subequations}
Using the data matrix in~\eqref{eq: Si} and the extended dictionary $\mathcal{Z}_i$ in~\eqref{eq: dictionary}, we proceed to construct the following matrix:
\begin{align}
	\PD \!=\! \begin{bmatrix}
		x_i (t_0) & x_i (t_0 \!+\! \tau) & \!\!\!\dots & x_i (t_0 \!+\! (\mathcal{T} \!-\! 1)\tau)\\
		\mathcal{M}_i(x_i (t_0)\!) & 	\!\!\!\mathcal{M}_i(x_i (t_0 \!+\! \tau)\!) & \!\!\!\dots & 	\!\!\!\!\!\mathcal{M}_i(x_i (t_0 \!+\! (\mathcal{T} \!-\! 1)\tau)\!) \!
	\end{bmatrix}\!\!. \label{eq: PD}
\end{align}

The following lemma establishes the data-driven representation of each nominal \Subsys\ $\Upsilon_i^\ast$. This formulation further enables the characterization of the unknown matrix $\mathcal{B}_i$, which is subsequently utilized in Section~\ref{subsec: ISM_data}.

\begin{lemma}[\textbf{Data-Based Closed-Loop Representation of $\Upsilon_i^\ast$}]\label{Lemma-rep}
	Given a nominal \Subsys\ $\Upsilon_i^\ast$, let $\mathcal{G}_i = [\mathcal{G}_{1_i} \;\; \mathcal{G}_{2_i}] \in \R^{\mathcal{T} \times z_i}$ and $\mathcal{Q}_i \in \R^{\mathcal{T} \times m_i}$ represent arbitrary matrices, with $\mathcal{G}_{1_i} \in \R^{\mathcal{T} \times n_i}$ and $\mathcal{G}_{2_i} \in \R^{\mathcal{T} \times (z_i - n_i)}$, that satisfy
	\begin{subequations}\label{eq: lemma_equalities}
		\begin{align}
			\I_{z_i} & = \PD \mathcal{G}_i, \label{eq: lemma_eq1}\\
			\mathcal{R}_i & \; = \mathbfcal{J}_i \mathcal{Q}_i, \label{eq: lemma_eq2}
		\end{align}
	\end{subequations}
	where $\mathcal{R}_i = [\bbzero_{m_i \times z_i} \;\; \I_{m_i}]^\top$ and $\mathbfcal{J}_i = [\PD^{\! \!\! \top} \;\; \Ii^{\! \! \top}]^\top$.
	Then, under the nominal local control law
	\begin{align}
		u_i^\ast = \mathbfcal{K}_i\mathcal{Z}_i(x_i), \label{eq: lemma_u}
	\end{align}
	where $\mathbfcal{K}_i = \Ii \mathcal{G}_i$,
	the data-based representation of $\mathcal{A}_i \mathcal{Z}_i(x_i) + \mathcal{B}_i u_i^\ast$ can be expressed as
	\begin{align}
		\mathcal{A}_i \mathcal{Z}_i(x_i) + \mathcal{B}_i u_i^\ast = \; & (\Sip - \mathcal{D}_i \Wi)\mathcal{G}_{1_i} x_i \notag\\
		&+ (\Sip - \mathcal{D}_i \Wi)\mathcal{G}_{2_i} \mathcal{M}_i(x_i). \label{eq: lemma_closed}
	\end{align}
	Moreover, the control input matrix $\mathcal{B}_i$ has the following data-based characterization:
	\begin{align}
		\mathcal{B}_i = (\Sip - \mathcal{D}_i \Wi) \mathcal{Q}_i  . \label{eq: lemma_B}
	\end{align}
\end{lemma}

\begin{proof}
	According to the collected data in~\eqref{eq: data_ct} and~\eqref{eq: PD}, as well as the dynamics of the nominal \Subsys\ $\Upsilon_i^\ast$, \emph{i.e.}, $\dot{x}_i  = \mathcal{A}_i \mathcal{Z}_i(x_i ) + \mathcal{B}_i u_i + \mathcal{D}_i w_i $, one has
	\begin{align}
		\Sip = \mathcal{A}_i \PD + \mathcal{B}_i \Ii + \mathcal{D}_i \Wi. \label{tmp1}
	\end{align}
By multiplying both sides of~\eqref{tmp1} by $\mathcal{Q}_i$ from the right, one has
	\begin{align*}
		\Sip \mathcal{Q}_i = \mathcal{A}_i \PD \mathcal{Q}_i + \mathcal{B}_i \Ii \mathcal{Q}_i + \mathcal{D}_i \Wi \mathcal{Q}_i.
	\end{align*}
	Consequently, considering~\eqref{eq: lemma_eq2}, we have
		\[
		\Sip \mathcal{Q}_i = \mathcal{A}_i \bbzero_{z_i \times m_i} + \mathcal{B}_i \I_{m_i} + \mathcal{D}_i \Wi \mathcal{Q}_i,
		\] or equivalently $\mathcal{B}_i = \Sip \mathcal{Q}_i - \mathcal{D}_i \Wi \mathcal{Q}_i$, which is in accordance with \eqref{eq: lemma_B}. At the same time, one has
	\begin{align*}
		&\mathcal{A}_i \mathcal{Z}_i(x_i ) + \mathcal{B}_i u_i^\ast\\
		& \overset{\eqref{eq: lemma_u}}{=} \mathcal{A}_i \mathcal{Z}_i(x_i ) + \mathcal{B}_i \Ii \mathcal{G}_i\mathcal{Z}_i(x_i ) = (\mathcal{A}_i + \mathcal{B}_i \Ii \mathcal{G}_i) \mathcal{Z}_i(x_i )\\
		& \! \overset{\eqref{eq: lemma_eq1}}{=} (\mathcal{A}_i \overbrace{\PD \mathcal{G}_i}^{\I_{z_i}} + \mathcal{B}_i \Ii \mathcal{G}_i) \mathcal{Z}_i(x_i ) = (\mathcal{A}_i \PD + \mathcal{B}_i \Ii) \mathcal{G}_i \mathcal{Z}_i(x_i )\\
		& \overset{\eqref{tmp1}}{=} (\Sip - \mathcal{D}_i \Wi) \mathcal{G}_i \mathcal{Z}_i(x_i ) \!=\! (\Sip - \mathcal{D}_i \Wi) [\mathcal{G}_{1_i} \;\; \mathcal{G}_{2_i}] \mathcal{Z}_i(x_i )\\
		& \, \overset{\eqref{eq: dictionary}}{=}  (\Sip - \mathcal{D}_i \Wi)\mathcal{G}_{1_i} x_i  + (\Sip - \mathcal{D}_i \Wi)\mathcal{G}_{2_i} \mathcal{M}_i(x_i ),
	\end{align*}
	thereby concluding the proof.
\end{proof}

\begin{remark}[\textbf{Data Richness}]\label{Remark-richness}
	A matrix $\mathcal{G}_i$ satisfying condition~\eqref{eq: lemma_eq1} exists if and only if $\PD$ has full row rank. Likewise, a matrix $\mathcal{Q}_i$ satisfying condition~\eqref{eq: lemma_eq2} exists if and only if $\operatorname{col}(\mathcal{R}_i)\subseteq \operatorname{col}(\mathbfcal{J}_i)$. Notice that the subsequent result in Theorem~\ref{Theorem-data-iss} only requires the existence of $\mathcal{G}_i$; thus, $\PD$ being of full row rank is necessary for applying Theorem~\ref{Theorem-data-iss}, which implies that $\mathcal{T}\geq z_i$ is the least possible number of collected samples required (see also Section~\ref{Section_feas}). For the simultaneous existence of $\mathcal G_i$ and $\mathcal Q_i$, a necessary and sufficient rank condition is
		\(
		\operatorname{rank}(\mathbfcal J_i)=z_i+m_i,
		\)
		implying that $\mathcal T\ge z_i+m_i$. Nevertheless, these lower bounds on the number of collected samples are only necessary sample-size requirements; in practice, the number of collected samples may need to be larger to satisfy the corresponding rank conditions.
\end{remark}

Now, by utilizing the data-based representation of $\mathcal{A}_i \mathcal{Z}_i(x_i ) + \mathcal{B}_i u_i^\ast$ established in Lemma \ref{Lemma-rep}, we present the following theorem as one of the key contributions of this work. This theorem provides conditions for constructing ISS Lyapunov functions and their local controllers for each nominal subsystem directly from data. We note that this theorem requires only the existence of a matrix $\mathcal G_i$ satisfying~\eqref{eq: lemma_eq1}.

\begin{theorem}[\textbf{Data-Driven ISS Lyapunov Functions}]\label{Theorem-data-iss}
	Consider an unknown nominal \Subsys\ $\Upsilon_i^\ast$, and the data-based representation of $\mathcal{A}_i \mathcal{Z}_i(x_i) + \mathcal{B}_i u_i^\ast$ in~\eqref{eq: lemma_closed}. If, for some fixed $\kappa_i \in \Rp$, there exist a constant $\mu_i \in \Rp$ and matrices $\mathcal{G}_{2_i} \in \R^{\mathcal{T} \times (z_i - n_i)}$, $\mathcal{Y}_i \in \R^{\mathcal{T} \times n_i}$, and $\Phi_i \in \R^{n_i \times n_i}$, where $\Phi_i \succ 0$, fulfilling
	\begin{subequations}\label{eq: thm-iss}
		\begin{align}
			(\Sip - \mathcal{D}_i \Wi) \mathcal{G}_{2_i} &= \bbzero_{n_i \times (z_i - n_i)}, \label{eq : thm1}\\
			\PD \mathcal{G}_{2_i} &= \begin{bmatrix}
				\bbzero_{n_i \times (z_i - n_i)}\\
				\I_{z_i - n_i}
			\end{bmatrix}\!\!, \label{eq : thm2}\\
			\PD \mathcal{Y}_i &= \begin{bmatrix}
				\Phi_i\\
				\bbzero_{(z_i - n_i) \times n_i}
			\end{bmatrix}\!\!,\label{eq : thm3}\\
			\mathcal{Y}_i^\top (\Sip - \mathcal{D}_i \Wi)^{\! \top}&+ (\Sip - \mathcal{D}_i \Wi) \mathcal{Y}_i + \mu_i \I_{n_i} \preceq - \kappa_i \Phi_i, \label{eq : thm4}
		\end{align}
	\end{subequations}
	then the control law $u_i^\ast = \mathbfcal{K}_i\mathcal{Z}_i(x_i)$, with $\mathbfcal{K}_i = \Ii [\mathcal{Y}_i \MPI \;\, \mathcal{G}_{2_i}]$, renders the nominal closed-loop \Subsys\ $\Upsilon_i^\ast$ ISS, where
	$\V_i(x_i) = x_i^\top \MPI x_i$, with $\MPI \coloneq \Phi_i^{-1}$, is the corresponding ISS Lyapunov function, with $\alpha_{1_i} = \lambda_{\min}(\MPI), \alpha_{2_i} = \lambda_{\max}(\MPI),$ and $\rho_i = \frac{\Vert \mathcal{D}_i \Vert^2}{\mu_i}$.
\end{theorem}

\begin{proof}
	We initiate the proof by showing the satisfaction of condition \eqref{eq: ISS-con1}. Considering the \emph{quadratic} form of the candidate ISS Lyapunov function $\V_i$, we get
	\begin{align*}
		\lambda_{\min}(\MPI) \vert x_i \vert^2 \leq \V_i(x_i) \leq \lambda_{\max}(\MPI) \vert x_i \vert^2,
	\end{align*}
	concluding the fulfillment of condition \eqref{eq: ISS-con1} with $\alpha_{1_i} = \lambda_{\min}(\MPI)$ and $\alpha_{2_i} = \lambda_{\max}(\MPI).$
	
	We now continue by showing that under the proposed conditions in~\eqref{eq: thm-iss}, condition~\eqref{eq: ISS-con2} is also met. Let us define $\mathcal{G}_{1_i} \coloneq \mathcal{Y}_i \MPI$, with $\MPI \coloneq \Phi_i^{-1}$. Subsequently, it is evident that conditions~\eqref{eq : thm2} and~\eqref{eq : thm3} together result in
	\begin{align*}
		\PD \begin{bmatrix}
			\mathcal{G}_{1_i} & \mathcal{G}_{2_i}
		\end{bmatrix} = \I_{z_i},
	\end{align*}
	which is in accordance with~\eqref{eq: lemma_eq1}. Consequently, based on the Lie derivative in~\eqref{eq: Lie derivative-subsystem}, one has
	\begin{align*}
		&\mathsf{L} \V_i(x_i)\\
		& = \partial_{x_i} \V_i(x_i) (\mathcal{A}_i \mathcal{Z}_i(x_i) + \mathcal{B}_i u_i^\ast + \mathcal{D}_i w_i)\\
		&= 2x_i^\top \MPI (\mathcal{A}_i \mathcal{Z}_i(x_i) + \mathcal{B}_i u_i^\ast + \mathcal{D}_i w_i)\\
		& \! \overset{\eqref{eq: lemma_closed}}{=} 2x_i^\top \MPI ((\Sip - \mathcal{D}_i \Wi)\mathcal{G}_{1_i} x_i  + (\Sip - \mathcal{D}_i \Wi)\mathcal{G}_{2_i} \mathcal{M}_i(x_i )\\
		& \hspace{1.6cm}+ \mathcal{D}_i w_i )\\
		& \!\! \overset{\eqref{eq : thm1}}{=} 2x_i^\top \MPI ((\Sip - \mathcal{D}_i \Wi)\mathcal{G}_{1_i} x_i + \mathcal{D}_i w_i )\\
		& = 2x_i^\top \MPI (\Sip - \mathcal{D}_i \Wi)\mathcal{G}_{1_i} x_i + 2 \overbrace{x_i^\top \MPI}^{a_i} \overbrace{\mathcal{D}_i w_i}^{b_i}.
	\end{align*}
	By utilizing the Cauchy--Schwarz inequality~\cite{bhatia1995cauchy}, which states that $a_i b_i \leq \vert a_i \vert \vert b_i \vert$ for any $a_i^\top , b_i \in \R^{n_i}$, and subsequently applying Young's inequality~\cite{young1912classes}, \emph{i.e.}, $\vert a_i \vert \vert b_i \vert \leq \frac{\mu_i}{2} \vert a_i \vert^2 + \frac{1}{2\mu_i} \vert b_i \vert^2$ for any $\mu_i \in \Rp$, one has
	\begin{align*}
		&\mathsf{L} \V_i(x_i)\\
		& \leq 2x_i^\top \MPI (\Sip - \mathcal{D}_i \Wi)\mathcal{G}_{1_i} x_i + \mu_i x_i^\top \MPI \MPI x_i + \frac{1}{\mu_i} \vert \mathcal{D}_i w_i \vert^2.
	\end{align*}
	Through expanding the term obtained and factorizing $x_i^\top \MPI$ from the left and $\MPI x_i$ from the right, and since $\mathcal{G}_{1_i} \coloneq \mathcal{Y}_i \MPI,$ we obtain
	\begin{align*}
		&\mathsf{L} \V_i(x_i)\\
		& \leq x_i^\top \MPI  \big[\overbrace{\MPI^{-1} \mathcal{G}_{1_i}^\top}^{\mathcal{Y}_i^\top} (\Sip - \mathcal{D}_i \Wi)^{\! \top} \!+\! (\Sip - \mathcal{D}_i \Wi) \overbrace{\mathcal{G}_{1_i} \MPI^{-1}}^{\mathcal{Y}_i } \\
		& \hspace{1.4cm} + \mu_i \I_{n_i} \big] \MPI x_i + \frac{1}{\mu_i} \Vert \mathcal{D}_i \Vert^2 \vert w_i \vert^2\\
		& = x_i^\top \MPI \big[\mathcal{Y}_i^\top (\Sip \!-\! \mathcal{D}_i \Wi)^{\! \top} \!+\! (\Sip \!-\! \mathcal{D}_i \Wi) \mathcal{Y}_i  \!+\! \mu_i \I_{n_i} \big] \MPI x_i\\
		& \hspace{.4cm}  + \frac{1}{\mu_i} \Vert \mathcal{D}_i \Vert^2 \vert w_i \vert^2\\
		& \! \overset{\eqref{eq : thm4}}{\leq} -\kappa_i \underbrace{x_i^\top \overbrace{\MPI \MPI^{-1}}^{\I_{n_i}} \MPI x_i}_{\V_i(x_i)} + \frac{1}{\mu_i} \Vert \mathcal{D}_i \Vert^2 \vert w_i \vert^2\\
		& = -\kappa_i \V_i(x_i) + \rho_i \vert w_i \vert^2,
	\end{align*}
	with $\rho_i = \frac{\Vert \mathcal{D}_i \Vert^2}{\mu_i},$ thereby concluding the proof.
\end{proof}

 \begin{remark}[\textbf{On Nominal Controller}]\label{rem:U_canc}
		Considering~\eqref{eq: lemma_closed} in Lemma~\ref{Lemma-rep}, it is evident that we decouple the data-based representation of each nominal closed-loop subsystem into its linear and nonlinear parts. Thus, if condition~\eqref{eq : thm1} is satisfied, the data-based representation of the nominal closed-loop subsystem includes $(\Sip - \mathcal{D}_i \Wi)\mathcal{G}_{1_i} x_i$, which is linear. To satisfy condition~\eqref{eq : thm1} (\emph{i.e.}, canceling the effect of matched nonlinearities), the nonlinear vector function $\mathcal{M}_i$ contained in the dictionary $\mathcal{Z}_i$ should be incorporated into the nominal controller, justifying the choice of the nominal control law $u_i^\ast = \mathbfcal{K}_i \mathcal{Z}_i(x_i)$.
\end{remark}

\subsection{Data-Driven ISM Control}\label{subsec: ISM_data}
In this subsection, we build upon the data-driven controller $u_i^\ast$, developed in the previous subsection, to design the discontinuous control law $u_{i}^{\mathrm{ISM}}$, as described in~\eqref{eq: ISM_U}, for each \Subsys\ $\Upsilon_i$ using data,  thereby eliminating the need for knowledge of the nominal dynamics required in classical ISM control~\cite{utkin1996integral}. This leads to the following novel data-driven ISM control.

By leveraging the data-based representation in~\eqref{eq: lemma_closed}, the dynamics of the transient function in~\eqref{eq: transient} can be reformulated as
\begin{align}
	\dot{\zeta}_i & = -\frac{\partial \sigma_{0_i}(x_i)}{\partial x_i} (\mathcal{A}_i \mathcal{Z}_i(x_i) + \mathcal{B}_i u_i^\ast + \mathcal{D}_i w_i)\notag\\
	& \! \overset{\eqref{eq: lemma_closed}}{=} -\frac{\partial \sigma_{0_i}(x_i)}{\partial x_i}((\Sip - \mathcal{D}_i \Wi) \mathcal{G}_i \mathcal{Z}_i(x_i) + \mathcal{D}_i w_i), \label{eq: transient_new}
\end{align}
with $\mathcal{Z}_i(x_i)$ as in~\eqref{eq: dictionary} and $\zeta_i(t_0) = - \sigma_{0_i}(x_i(t_0))$.

We now focus on designing the controller $u_{i}^{\mathrm{ISM}}$ in~\eqref{eq: ISM_U} for each $\Upsilon_i$ using data, ensuring that the controlled subsystem remains robust against the matched perturbation $\mathcal{E}_i$.
Note that the perturbation vector $\mathcal{E}_i(x_i, t)$ satisfies the matching condition $\mathcal{E}_i(x_i, t) \! = \! \mathcal{B}_i \gamma_i(x_i, t)$ for all $(x_i, t) \! \in \! \R^{n_i} \! \times \! \Rpz$, where $\gamma_i(x_i, t) \! \in \! \R^{m_i}$ and $\gamma_i$ is continuous. It is assumed that 
\begin{subequations}\label{eq: Gamma3}
	\begin{align}
		\gamma_i(x_i, t) \in \Gamma_i, \quad \forall(x_i, t) \in \R^{n_i} \times \Rpz, \label{eq: Gamma2}
	\end{align}
	where $\Gamma_i$ is a compact set that contains the origin. Furthermore, we assume that $\Gamma_i^{\sup}$, defined as
	\begin{align}
		\Gamma_i^{\sup} = \sup_{\gamma_i \in \Gamma_i} \vert \gamma_i \vert, \label{eq: Gamma}
	\end{align}
\end{subequations}
is known and available. Indeed, considering Assumption~\ref{Assumption: bound_disturbance}, one can choose $\Gamma_i^{\sup} = \Vert (\mathcal{B}_i^\top \mathcal{B}_i)^{-1} \mathcal{B}_i^\top \Vert \bar{\mathbfcal E}_i$, with $\mathcal{B}_i$ as in~\eqref{eq: lemma_B}. 

Here, the term $\sigma_{0_i}$ in the integral sliding variable~\eqref{eq: integral sliding variable} is chosen as a linear combination of the state variables, expressed as $\sigma_{0_i}(x_i) = \mathcal{C}_i x_i$, where $\mathcal{C}_i = \frac{\partial \sigma_{0_i}(x_i)}{\partial x_i} \in \R^{m_i \times n_i}$. Note that the matrix $\mathcal{C}_i$ should be designed such that $\mathcal{C}_i \mathcal{B}_i \succ 0$. However, since $\mathcal{B}_i$ is assumed to be unknown, we substitute it with its data-based characterization obtained in~\eqref{eq: lemma_B}. Hence, the matrix $\mathcal{C}_i$ should be designed such that $\mathcal{C}_i (\Sip - \mathcal{D}_i \Wi) \mathcal{Q}_i \succ 0$. 

In the following lemma, we propose a condition under which the integral sliding manifold $\sigma_i = \bbzero_{m_i}$ is maintained for all $t \in \Rpz$, thereby consistently rejecting the perturbation; that is, maintaining the robustness properties of the model-based counterpart in~\cite{utkin1996integral}.

\begin{lemma}[\textbf{Designing $\Theta_i$}]\label{Lemma: rho}
	Consider a \Subsys\ $\Upsilon_i$ as in~\eqref{eq: final_subsys}, controlled by $u_i = u_i^\ast + u_{i}^{\mathrm{ISM}}$, where $u_i^\ast = \mathbfcal{K}_i \mathcal{Z}_i(x_i)$, with $\mathbfcal{K}_i = \Ii \mathcal{G}_i$ for some $\mathcal{G}_i$ satisfying~\eqref{eq: lemma_eq1}, $u_{i}^{\mathrm{ISM}}$ is given by \eqref{eq: ISM_U}, and the integral sliding variable is chosen as in~\eqref{eq: integral sliding variable}, with $\sigma_{0_i}(x_i) = \mathcal{C}_i x_i$ and $\zeta_i$ as in~\eqref{eq: transient_new}. Let~\eqref{eq: Gamma3} hold, and suppose that there exists a matrix $\mathcal Q_i$ satisfying~\eqref{eq: lemma_eq2}, and the matrix $\mathcal{C}_i$ is designed such that $\mathcal{C}_i (\Sip - \mathcal{D}_i \Wi) \mathcal{Q}_i \succ 0$. If the control gain $\Theta_i \in \Rp$ is selected such that
	\begin{align}
		\Theta_i >\Gamma_i^{\sup}, \label{eq: Theta}
	\end{align}
	then a  sliding mode $\sigma_i = \bbzero_{m_i}$  is enforced for all $t \in \Rpz$.
\end{lemma}

\begin{proof}
	Consider a Lyapunov function $\mathbb{V}_{\! i} : \R^{m_i} \to \Rpz$ defined as
		\begin{align*}
			\mathbb{V}_{\! i} = \frac{1}{2} \sigma_i^\top  {\big(\mathcal{C}_i (\Sip - \mathcal{D}_i \Wi) \mathcal{Q}_i\big)^{-1}} \sigma_i \overset{\eqref{eq: lemma_B}}{=}  {\frac{1}{2} \sigma_i^\top (\mathcal{C}_i \mathcal{B}_i)^{-1} \sigma_i,}
		\end{align*}
		where $\mathcal{C}_i \mathcal{B}_i = (\mathcal{C}_i \mathcal{B}_i)^\top \succ 0$.
		By exploiting~\eqref{eq: transient_new}, one has
		\begin{align}\notag
			\dot{\mathbb{V}}_{\! i} & = \sigma_i^\top  {(\mathcal{C}_i \mathcal{B}_i)^{-1}} \dot{\sigma}_i = \sigma_i^\top  {(\mathcal{C}_i \mathcal{B}_i)^{-1}} (\dot{\sigma}_{0_i} + \dot{\zeta}_i)\\\notag
			& \! \overset{\eqref{eq: transient_new}}{=} \! \sigma_i^\top \!  {(\mathcal{C}_i \mathcal{B}_i)^{-1}} \big(\mathcal{C}_i \dot{x}_i \! - \! \mathcal{C}_i (\! (\Sip \! - \! \mathcal{D}_i \Wi) \mathcal{G}_i \mathcal{Z}_i(x_i) \! + \! \mathcal{D}_i w_i) \! \big)\\\notag
			& \!\!\!\! \overset{\eqref{eq: final_subsys},\, \eqref{eq: lemma_closed}}{=} \! \sigma_i^\top \!  {(\mathcal{C}_i \mathcal{B}_i)^{-1}} \big(\mathcal{C}_i ( \mathcal{A}_i \mathcal{Z}_i(x_i) + \mathcal{B}_i (u_i^\ast + u_{i}^{\mathrm{ISM}}) + \mathcal{D}_i w_i  \notag\\
			& \hspace{.5cm} + \mathcal{E}_i (x_i, t)) - \mathcal{C}_i (\mathcal{A}_i \mathcal{Z}_i(x_i) + \mathcal{B}_i u_i^\ast + \mathcal{D}_i w_i)\big).\label{New6}
		\end{align}
		Since $\mathcal{E}_i(x_i, t) = \mathcal{B}_i \gamma_i(x_i, t)$  for all $(x_i, t) \in \R^{n_i} \times \Rpz$, and as the common terms in~\eqref{New6} cancel each other, utilizing the Cauchy--Schwarz inequality, for $\sigma_i\neq\bbzero_{m_i}$, one gets
		\begin{align*}
			\dot{\mathbb{V}}_{\! i} & = \sigma_i^\top \overbrace{ {(\mathcal{C}_i \mathcal{B}_i)^{-1}}(\mathcal{C}_i \mathcal{B}_i)}^{ {\I_{m_i}}} (u_{i}^{\mathrm{ISM}} + \gamma_i(x_i, t)) \\ & \! \overset{\eqref{eq: ISM_U}}{=} \sigma_i^\top \! \big(-\Theta_i  \, \frac{\sigma_i}{\vert \sigma_i \vert} +  \gamma_i(x_i, t) \big) = -\Theta_i \, \frac{\sigma_i^\top \sigma_i}{\vert \sigma_i \vert} +  \sigma_i^\top \gamma_i (x_i, t) \\
			& \!\! \overset{\eqref{eq: Gamma}}{\leq} -\Theta_i  \vert \sigma_i \vert + \Gamma_i^{\sup}   \vert \sigma_i \vert = -\eta_i \vert \sigma_i \vert,
		\end{align*}
		with $\eta_i = \Theta_i - \Gamma_i^{\sup}$.
		Now, it is clear that $\eta_i$ is positive under the satisfaction of condition~\eqref{eq: Theta}, rendering $\dot{\mathbb{V}}_{\! i} < 0$ for all $\sigma_i\neq\bbzero_{m_i}$. Since the so-called \emph{reachability condition} holds, one can conclude that there exists $\bar{t}>0$ such that $\sigma_i({x_i(t)},t)=\bbzero_{m_i}$ for any $t\geq \bar{t}$. Moreover, since at $t_0 = 0$, we have $\sigma_i({x_i(t_0)}, t_0)=\bbzero_{m_i}$ by construction (according to \eqref{eq: integral sliding variable}), the preceding dissipation inequality ensures that $\mathbb{V}_{\! i}$ cannot increase from its zero initial value. Hence, the integral sliding manifold is invariant from the initial time instant, \emph{i.e.}, $\sigma_i(x_i(t),t)=\bbzero_{m_i}$ for all $t\ge 0$ along the closed-loop trajectory, which concludes the proof.
\end{proof}

With Lemma~\ref{Lemma: rho} proposed, we now continue with presenting the following theorem, which proves that the motion equation of each \Subsys\ $\Upsilon_i$, $i \in \{1, \ldots, N\}$, in sliding mode is equal to its nominal dynamics controlled by $u_i^\ast$.

\begin{theorem}[\textbf{Motion Equation in Sliding Mode}]\label{Theorem: Filippov}
	Consider a \Subsys\ $\Upsilon_i$ as in~\eqref{eq: final_subsys}, and suppose that the SDP~\eqref{eq: thm-iss} is feasible, that there exists a matrix $\mathcal Q_i$ satisfying~\eqref{eq: lemma_eq2}, and that $\mathcal{C}_i$ is designed such that $\mathcal{C}_i (\Sip - \mathcal{D}_i \Wi) \mathcal{Q}_i \succ 0$. Given the control law $u_i = u_i^\ast + u_{i}^{\mathrm{ISM}}$, with $u_i^\ast$ as in~\eqref{eq: lemma_u}, where $\mathbfcal{K}_i$ is constructed as in Theorem~\ref{Theorem-data-iss}, $u_{i}^{\mathrm{ISM}}$ as in~\eqref{eq: ISM_U}, with $\Theta_i$ as in~\eqref{eq: Theta}, and the integral sliding variable $\sigma_i$ as in~\eqref{eq: integral sliding variable}, with $\zeta_i$ as in~\eqref{eq: transient_new}, if the sliding mode $\sigma_i = \bbzero_{m_i}$ is enforced for all $t \in \Rpz$, then the subsystem motion in sliding mode is
	\begin{equation}
		\dot x_i = \mathcal{A}_i \mathcal{Z}_i(x_i) + \mathcal{B}_i u_i^\ast + \mathcal{D}_i w_i, \label{eq: ISSsys}
	\end{equation}
	which is ISS under $u_i^\ast$ in~\eqref{eq: lemma_u}.
\end{theorem}

\begin{proof}
	By applying Lemma~\ref{Lemma: rho}, it is ensured that the ISM is $\sigma_i = \bbzero_{m_i}$ for all $t \in \Rpz$, which also implies that $\dot{\sigma}_i = \bbzero_{m_i}$ in Filippov's sense~\cite{filippov2013differential}. The latter enables the computation of the so-called \emph{equivalent control}, denoted as $u_{i_\text{eq}}^{\mathrm{ISM}}$, which is expressed as
	\[
	\mathcal{C}_i \mathcal{B}_i (u_{i_\text{eq}}^{\mathrm{ISM}} + \gamma_i) = \bbzero_{m_i}.
	\]
	Under the assumption that $\mathcal{C}_i \mathcal{B}_i\succ 0$, it follows that $u_{i_\text{eq}}^{\mathrm{ISM}} = -\gamma_i$. By substituting this equivalent control into the \Subsys\ $\Upsilon_i$ in~\eqref{eq: final_subsys}, which is now controlled by $u_i = u_i^\ast + u_{i_\text{eq}}^{\mathrm{ISM}}$, one can conclude that the nominal dynamics in~\eqref{eq: ISSsys} are obtained as the motion equation of each \Subsys\ $\Upsilon_i$, $i \in \{1, \ldots, N\}$, in sliding mode, which, according to Theorem~\ref{Theorem-data-iss}, are ISS under $u_i^\ast$ in~\eqref{eq: lemma_u}, thereby completing the proof.
\end{proof}

We note that, in the proposed framework, the interconnection term \(\mathcal D_iw_i\) is not treated as an uncertainty to be rejected by the ISM controller; rather, it is part of the nominal subsystem dynamics. In particular, the effect of  \(\mathcal D_iw_i\) is explicitly reflected in the ISM design through the data-driven transient function in~\eqref{eq: transient_new}, which contains the data-based representation of the nominal closed-loop dynamics. Consequently, when differentiating the sliding variable, the interconnection term cancels from \(\dot\sigma_i\), leaving only the ISM input and the matched external perturbation. Hence, the ISM gain condition depends only on the matched perturbation bound; the effect of the interconnection term is captured through the ISS gains, which are subsequently employed in Section~\ref{Section: Compositional} to formulate the compositional small-gain condition for the overall network.

 {\begin{remark}[\textbf{Unmatched External Perturbations}]\label{rem:unmatched}
		This work considers matched external perturbations, as commonly assumed in the literature, since this ensures that the system remains insensitive to such disturbances once the trajectories reach the sliding manifold. While unmatched perturbations are beyond the scope of the present work, it is worth noting that their effect can be mitigated following~\cite{rubagotti2011unmatched}, \emph{i.e.}, by choosing $\mathcal{C}_i = \mathcal{B}_i^\top$ (cf.~\eqref{eq: lemma_B}), thereby rejecting matched perturbations and avoiding amplification of unmatched components.
\end{remark}}

 \begin{remark}[\textbf{Integral Sliding Surface}]\label{rem:whyISM}
		An integral sliding surface is adopted to ensure robustness from the initial time instant. Unlike classical linear surfaces, the integral term eliminates the reaching phase, during which trajectories are sensitive to perturbations, and guarantees that the system evolves entirely in sliding mode~\cite{utkin1996integral}. This feature is particularly advantageous in our setting, as it enables the use of data-driven ISS Lyapunov functions for the nominal subsystems and facilitates the construction of a CLF for the nominal network by ensuring rejection of matched perturbations from the outset.
\end{remark}

With both components of each subsystem's controller designed using data, the following section presents our compositional condition, grounded in small-gain reasoning, to construct a CLF for the nominal network and its associated controller. This approach utilizes ISS Lyapunov functions of the individual nominal subsystems obtained from data. We subsequently extend this result to perturbed networks, ensuring that the closed-loop network is GAS at the origin.

\section{Compositional Framework}\label{Section: Compositional}
We analyze the nominal interconnected network $\Upsilon^\ast = \mathscr{N}(\Upsilon_1^\ast, \ldots, \Upsilon_N^\ast)$ by formulating a small-gain compositional condition and constructing a CLF derived from the ISS Lyapunov functions of its nominal subsystems. To achieve this, we define 
\begin{align}\label{new11}
	\mathcal{H} \coloneq \mathsf{diag}(\kappa_1, \dots, \kappa_N), \; \hat{\varrho} \coloneq \{\hat{\rho}_{ij}\}, \; \hat{\rho}_{ij} = \begin{cases}
			\frac{\rho_i}{\alpha_{1_j}}, & j \in N_i,\\
			0, &  j \notin N_i,
	\end{cases}
\end{align}
with $\hat{\rho}_{ii} = 0$ for all $i \in \{1, \dots, N\}$.

The following theorem establishes the conditions under which a CLF for an unknown nominal network can be constructed using the data-driven ISS Lyapunov functions of its nominal subsystems.

\begin{theorem}[\textbf{Compositional Results}]\label{Theorem: comp}
	Consider a nominal interconnected network $\Upsilon^\ast = \mathscr{N}(\Upsilon_1^\ast, \ldots, \Upsilon_N^\ast),$ consisting of $N \in \Np$ nominal subsystems $\Upsilon_i^\ast$. Assume that each nominal subsystem $\Upsilon_i^\ast$ is equipped with a data-driven ISS Lyapunov function $\V_i$, as constructed in Theorem~\ref{Theorem-data-iss}. If the following condition holds:
	\begin{align}
		\One_N^\top (-\mathcal{H} + \hat{\varrho}) \coloneq [\Xi_1 ; \dots ; \Xi_N]^\top < \bbzero_{1 \times N}, \label{eq: comp con}
	\end{align}
	or equivalently, $\Xi_i < 0$ for all $i \in \{1, \dots, N\}$, then $u^\ast \!=\! [u_1^\ast;\dots; u_N^\ast]$, where $u_i^\ast = \mathbfcal{K}_i \mathcal{Z}_i(x_i)$, $i \in \{1, \ldots, N\}$, with $\mathbfcal{K}_i = \Ii \begin{bmatrix}
			\mathcal{Y}_i \MPI & \mathcal{G}_{2_i}
		\end{bmatrix}\!$, renders the origin of the nominal closed-loop network GAS, where
	\begin{align}\label{New9}
		\V(x) \coloneq \sum_{i=1}^N \V_i(x_i) = \sum_{i=1}^N x_i^\top \MPI x_i
	\end{align}
	is the corresponding CLF, with
	\begin{align*}
		\kappa & \coloneq -\Xi, \qquad \qquad~ \max_{1 \leq i \leq N} \Xi_i \leq \Xi < 0, \\
		\alpha_1 & \coloneq \min_{1 \leq i \leq N} \{\alpha_{1_i}\}, \quad \alpha_2 \coloneq \max_{1 \leq i \leq N} \{\alpha_{2_i}\}.
	\end{align*}
\end{theorem}

\begin{proof}
	We begin by demonstrating that condition~\eqref{eq: clf-con1-network} is fulfilled. Based on condition~\eqref{eq: ISS-con1}, it follows that
	\begin{align*}
		\vert x \vert^2 = \sum_{i = 1}^{N}\vert x_i \vert^2 \overset{\eqref{eq: ISS-con1}}{\leq} \sum_{i = 1}^{N} \frac{1}{\alpha_{1_i}}\V_i(x_i) \leq \delta \sum_{i = 1}^{N} \V_i(x_i) \overset{\eqref{New9}}{=} \delta \V(x),
	\end{align*}
	where $\delta = \max_{1 \leq i \leq N}\{\frac{1}{\alpha_{1_i}}\}$. Then, through selecting $\alpha_1 = \frac{1}{\delta} = \min_{1 \leq i \leq N}\{\alpha_{1_i}\}$, one gets
		\begin{align*}
			\alpha_1\vert x \vert^2 \leq \V(x).
		\end{align*}
		Similarly, from the upper bound in~\eqref{eq: ISS-con1}, we have
		\begin{align*}
			\V(x) \!\overset{\eqref{New9}}{=}\! \sum_{i = 1}^{N} \V_i(x_i) \! \overset{\eqref{eq: ISS-con1}}{\leq} \!\sum_{i = 1}^{N} \!\alpha_{2_i} \vert x_i \vert^2 \leq \alpha_2\! \sum_{i = 1}^{N} \!\vert x_i \vert^2 \!= \!\alpha_2 \vert x \vert^2,
		\end{align*}
		with $\alpha_2=\max_{1 \leq i \leq N}\{\alpha_{2_i}\}$, indicating that condition \eqref{eq: clf-con1-network} is met with $\alpha_1 = \min_{1 \leq i \leq N}\{\alpha_{1_i}\}$ and $\alpha_2=\max_{1 \leq i \leq N}\{\alpha_{2_i}\}$.
		
		Next, we demonstrate that condition~\eqref{eq: clf-con2-network} is also satisfied. Utilizing condition~\eqref{eq: ISS-con2} and the compositional condition $\One_N^{\top}(-\mathcal{H} + \hat{\varrho}) < \bbzero_{1 \times N}$, and by defining
		\begin{align}\label{new12}
			\overline{\V}(x) \coloneq [\V_1( x_1) ; \dots ; \V_N(x_N)],
		\end{align}
		we obtain the subsequent chain of inequalities:
		\begin{align*}
			&\mathsf{L} \V(x) \\
			& \! \overset{\eqref{New9}}{=} \mathsf{L}\sum_{i=1}^{N} \V_i(x_i) \! = \! \sum_{i=1}^{N} \mathsf{L} \V_i(x_i) \! \overset{\eqref{eq: ISS-con2}}{\leq} \sum_{i=1}^{N} \! \big(\!- \! \kappa_i \V_i(x_i)\! + \! \rho_i\vert w_i\vert^2\big)\\
			&  \overset{\eqref{eq: internal-connection}}{=}     \sum_{i=1}^{N}  \big(  -  \kappa_i \V_i(x_i)  +    \sum_{j \in N_i}  \rho_i\vert w_{ij}\vert^2\big)     \\
			& \overset{\eqref{eq: internal-connection}}{=}     \sum_{i=1}^{N}  \big(  -  \kappa_i \V_i(x_i)  +    \sum_{j \in N_i}  \rho_i\vert x_j\vert^2\big)\\
			& \!\! \overset{\eqref{eq: ISS-con1}}{\leq}     \sum_{i=1}^{N}  \big(  -  \kappa_i \V_i(x_i)  +    \sum_{j \in N_i}  \frac{\rho_i}{\alpha_{1_j}}  \V_j(x_j)  \big)    \\
			& \! \overset{\eqref{new11}}{=}    \sum_{i=1}^{N}  \big(  -  \kappa_i \V_i(x_i)  +    \sum_{\substack{j=1 \\ j\neq i}}^{N}  \hat \rho_{i j}  \V_j(x_j)  \big)\\
			& \! \overset{\eqref{new11}}{=} \One_N^\top(-\mathcal H \! + \! \hat \varrho)[\V_1(x_1); \dots; \V_N(x_N)]\overset{\eqref{new12}}{=} \One_N^\top(-\mathcal H \! + \! \hat \varrho) \overline{\V}(x)\\
			& \! \overset{\eqref{eq: comp con}}{=} \sum_{i = 1}^N \Xi_i \V_i(x_i) \leq \Xi \sum_{i = 1}^N \V_i(x_i) \overset{\eqref{New9}}{=} -\kappa \V(x),
		\end{align*}
		where $\max_{1 \leq i \leq N} \Xi_i \leq \Xi < 0$ and $\kappa \coloneq -\Xi > 0$,
		implying that condition~\eqref{eq: clf-con2-network} is met, completing the proof.
\end{proof}

We now move on to extending the compositional result to the interconnected network $\Upsilon$ in the presence of external perturbations, as proposed in the following corollary.

\begin{algorithm}[t!]
		\caption{Data-driven design of the network CLF and its corresponding controller}\label{Alg:1}
		\begin{center}
			\begin{algorithmic}[1]
				\REQUIRE 
				An available dictionary $\mathcal{Z}_i$, {$\Gamma_i^{\sup}$}
                \FOR{$i = 1, \dots, N$}
                \STATE Collect data and form $\Ii$, $\Si$, $\Wi$, and $\Sip$ as in~\eqref{eq: data_ct} \label{STEP2}
                \STATE Given $\mathcal{Z}_i$ as in~\eqref{eq: dictionary}, construct $\PD$ as in~\eqref{eq: PD}
                \STATE For a fixed $\kappa_i$, solve \eqref{eq: thm-iss}, and design $\mu_i$, $\mathcal{G}_{2_i}$, $\mathcal{Y}_i$, and $\Phi_i \succ 0$
                \STATE Given $\MPI=\Phi_i^{-1}$, set $\mathcal G_{1_i} = \mathcal Y_i \MPI$ and $\mathcal G_i = [\mathcal G_{1_i} \; \mathcal G_{2_i}]$, and compute $\alpha_{1_i}$, $\alpha_{2_i}$, and $\rho_i$ as in Theorem~\ref{Theorem-data-iss}
                \STATE Construct $\V_i(x_i) = x_i^\top \MPI x_i$ as an ISS Lyapunov function for $\Upsilon_i^\ast$ with its local control law $u_i^\ast = \Ii \begin{bmatrix}
	               \mathcal{Y}_i \MPI & \mathcal{G}_{2_i}
	               \end{bmatrix}\mathcal{Z}_i(x_i)$ as in Theorem~\ref{Theorem-data-iss}
	            \STATE Compute $\mathcal{Q}_i$ from~\eqref{eq: lemma_eq2}
                \STATE Determine $\mathcal{C}_i$ such that $\mathcal{C}_i (\Sip - \mathcal{D}_i \Wi) \mathcal{Q}_i \succ 0$
                \STATE Construct $\sigma_{0_i}(x_i) = \mathcal{C}_i x_i$ and compute $\zeta_i$ from~\eqref{eq: transient_new}
                \STATE Design $\Theta_i$ according to \eqref{eq: Theta}
                \STATE Obtain each local integral sliding variable $\sigma_i$ as in~\eqref{eq: integral sliding variable} and the local ISM controller $u_{i}^{\mathrm{ISM}}$ as in~\eqref{eq: ISM_U}
		          \ENDFOR
                \IF{compositional condition~\eqref{eq: comp con} is met}
                \STATE $\V(x) = \sum_{i=1}^N x_i^\top \MPI x_i$ is a CLF for the nominal network $\Upsilon^\ast$, with $ \sigma \!=\! [\sigma_1;\dots; \sigma_N]$ being the integral sliding variable for the perturbed network $\Upsilon$, where $\sigma_i$ is constructed according to~\eqref{eq: integral sliding variable}, and with $ u \!=\! [u_1;\dots; u_N]$  being its global stabilizer, where $u_i = u_i^\ast + u_{i}^{\mathrm{ISM}}$, according to Corollary~\ref{Corollary: ext}
                \ELSE
                \STATE Return to Step \ref{STEP2}; collect a different set of data, potentially with a longer horizon, and proceed with the remaining steps
                \ENDIF
				\ENSURE
				Upon successful termination, CLF $\V(x) = \sum_{i=1}^N x_i^\top \MPI x_i$ for the nominal network $\Upsilon^\ast$, integral sliding variable $ \sigma \!=\! [\sigma_1;\dots; \sigma_N]$ for the perturbed network $\Upsilon$, with $\sigma_i$ as in~\eqref{eq: integral sliding variable}, and its global stabilizer $ u \!=\! [u_1;\dots; u_N],$ with $u_i = u_i^\ast + u_{i}^{\mathrm{ISM}}$ 
			\end{algorithmic}
		\end{center}
	\end{algorithm}

\begin{corollary}[\textbf{GAS of Network}]\label{Corollary: ext}
	Consider an interconnected network $\Upsilon = \mathscr{N}(\Upsilon_1, \ldots, \Upsilon_N)$ consisting of $N \in \Np$ \Subsys\ $\Upsilon_i$, controlled by $u = [u_1;\dots;u_N]$, where $u_i = u_i^\ast + u_i^{\mathrm{ISM}}$ for all $i \in \{1,\ldots,N\}$. Suppose that the hypotheses of Theorem~\ref{Theorem: comp} hold for the nominal closed-loop network $\Upsilon^\ast$. Moreover, suppose that, for each $i \in \{1,\ldots,N\}$, the hypotheses of Lemma~\ref{Lemma: rho} hold, and the ISM gain $\Theta_i$ is selected according to~\eqref{eq: Theta}. Then, for the integral sliding variable $\sigma = [\sigma_1;\dots;\sigma_N]$, with $\sigma_i$ as in~\eqref{eq: integral sliding variable}, the sliding manifold $\sigma=\bbzero_m$ is enforced for all $t \in \Rpz$. Consequently, the network motion in sliding mode coincides with that of the nominal controlled network $\Upsilon^\ast$, and the origin of the closed-loop interconnected network is GAS.
\end{corollary}

\begin{proof}
	The proof directly follows from Lemma~\ref{Lemma: rho}, Theorem~\ref{Theorem: Filippov}, and Theorem~\ref{Theorem: comp}. More precisely, since each $\Theta_i$ satisfies condition~\eqref{eq: Theta}, it follows that $\sigma_i = \bbzero_{m_i}$ is enforced for all $t \in \Rpz$ according to Lemma~\ref{Lemma: rho}. Since $ \sigma \!=\! [\sigma_1;\dots; \sigma_N],$ we readily get $\sigma = \bbzero_{m}$ for all $t \in \Rpz$. Therefore, according to Theorem~\ref{Theorem: Filippov}, it is clear that each subsystem's motion equation is equivalent to its nominal controlled dynamics. As a result, one can conclude that the result of Theorem~\ref{Theorem: comp} extends to the network with perturbations, ensuring that the closed-loop network $\Upsilon$ is GAS at the origin.
\end{proof}

 \begin{remark}[\textbf{Comparison with Existing Compositional Frameworks}]\label{rem:comp}
		We acknowledge that our compositional setting is inspired by prior compositional results such as~\cite{mironchenko2023input} and~\cite{dashkovskiy2010small}. While these studies employ $\max$-type small-gain conditions with more general, possibly non-smooth, ISS Lyapunov functions, they assume that subsystem ISS Lyapunov functions and the corresponding gains are known a priori. In contrast, in our framework, the ISS Lyapunov functions and quantities~$\rho_i$ and~$\alpha_{1_j}$ are obtained from data via Theorem~\ref{Theorem-data-iss}, making the proposed compositional scheme inherently data-driven. Furthermore, while~\cite{dashkovskiy2010small} ensures ISS and guarantees GAS only in the absence of external disturbances, our framework achieves GAS of the closed-loop interconnected network under nonvanishing external perturbation~$\mathcal{E}$ (cf. Corollary~\ref{Corollary: ext}).
\end{remark}

\begin{remark}[\textbf{On Data Collection}]\label{remark:noise}
	Lemma~\ref{Lemma-rep} requires data collected from nominal subsystems, which may introduce certain restrictions. This assumption, however, reflects a practical setting in which data are gathered in controlled environments, such as high-fidelity simulators or digital twins, prior to deployment, while the proposed ISM controller is designed to reject perturbations during real-world operation.
	If the data were collected from perturbed subsystems, three main challenges would arise: \emph{(i)} the transient function in~\eqref{eq: transient_new} would no longer exactly cancel the nominal closed-loop dynamics, leaving an unmodeled residual in the sliding dynamics; \emph{(ii)} condition~\eqref{eq : thm1} would no longer generally guarantee exact cancellation of the true nonlinear components associated with \(\mathcal M_i\), and a robust formulation may instead lead to approximate cancellation or residual minimization, potentially weakening the stability guarantees from ISS to \emph{local} or \emph{practical} ISS; and \emph{(iii)} the data-based characterization of \(\mathcal B_i\) may no longer recover the true input matrix, thereby affecting the design of \(\mathcal C_i\) and the condition \(\mathcal C_i\mathcal B_i \succ 0\). Addressing these challenges lies beyond the scope of the current work and is therefore left for future research.
\end{remark}

We present Algorithm~\ref{Alg:1}, which, assuming the feasibility of the SDP~\eqref{eq: thm-iss} and the existence of a matrix $\mathcal Q_i$ satisfying~\eqref{eq: lemma_eq2} for every subsystem, outlines all the necessary steps for the data-driven design of the local ISS Lyapunov functions and associated controllers, local ISM controllers, and a CLF for the nominal network together with the integral sliding variable and controller for the perturbed network, with GAS of the closed-loop network guaranteed when~\eqref{eq: comp con} holds.

\section{Discussion}\label{sec:discussion}
\subsection{Scalability Analysis}
 Here, we evaluate the scalability improvement over the centralized approach, \emph{i.e.}, applying Theorem~\ref{Theorem-data-iss} to the entire network, in terms of \emph{(i)} the number of decision variables and \emph{(ii)} rough asymptotic time and space complexities.

 As outlined in Algorithm~\ref{Alg:1}, our compositional framework solves the SDP~\eqref{eq: thm-iss} independently for each nominal subsystem; hence, excluding the scalar decision variable $\mu_i$, the number of scalar decision variables in this case is
	\begin{align}
		\mathcal{N}_{\textrm{sv}} \!
		=\! \tfrac{n_i (n_i  +  1)}{2}\!
		+\! \mathcal{T} (z_i \! - \! n_i)\!
		+ \! \mathcal{T} n_i \!
		= \! \tfrac{n_i (n_i  +  1)}{2} \! + \! \mathcal{T} z_i,\label{eq:1}
	\end{align}
	where $\tfrac{n_i (n_i + 1)}{2}$ corresponds to the \emph{symmetric} matrix $\Phi_i \in \R^{n_i \times n_i}$. For a clearer comparison, consider the special case of a \emph{homogeneous} network consisting of $N$ subsystems, where $N\mathcal{T}$ samples are used for a comparable centralized formulation. If one were to solve the problem in a centralized manner, the number of decision variables would be
	\begin{align}
		\mathcal{N}_{\textrm{nv}} 
		&= \frac{Nn_i (Nn_i + 1)}{2} + (N \mathcal{T})(Nz_i - Nn_i) + (N\mathcal{T})(Nn_i) \notag\\
		&= \frac{N^2n_i^2 + Nn_i}{2} + N^2\mathcal{T} z_i,\label{eq:tmp1}
	\end{align}
	which, for large values of $N$ (\emph{i.e.}, large-scale networks), can be approximated as
	\begin{align}
		\mathcal{N}_{\textrm{nv}} 
		\sim N^2 \Big(\frac{1}{2}n_i^2 + \mathcal{T} z_i\Big) \!. \label{eq:2}
	\end{align}
	From~\eqref{eq:1} and~\eqref{eq:2}, it is evident that while the number of decision variables has quadratic growth in \(N\) in the centralized approach (\emph{i.e.}, $\mathcal{N}_{\textrm{nv}} \asymp N^2$, meaning that $\mathcal{N}_{\textrm{nv}}$ and $N^2$ grow at the same rate, up to constant factors), it is independent of the number of subsystems $N$ in our compositional scheme for each subsystem, which is particularly advantageous for large-scale networks; the aggregate number of decision variables in our framework grows linearly with $N$. Furthermore, we note that the number of decision variables in the SDP~\eqref{eq: thm-iss} does not depend on the number of inputs or disturbances.

 To further study the scalability improvement, another important aspect of the SDP~\eqref{eq: thm-iss} is the number of constraints. More precisely, the SDP~\eqref{eq: thm-iss} has $\mathrm{c}_i = z_i^2 - n_i^2 + z_i n_i$ scalar equality constraints, in addition to two $n_i \times n_i$ LMIs. According to \emph{rough} time and memory complexity estimates for dense primal-dual interior-point methods~\cite{nesterov1994interior, vandenberghe1996semidefinite}, we have
	\begin{align}
		\begin{split}\label{eq:3}
			\mathsf{C}_i^t & =  \mathcal{O}((\mathrm{c}_i^3 + {\mathcal{N}_{\textrm{sv}}^2} n_i^2 + {\mathcal{N}_{\textrm{sv}}} n_i^3 + \mathcal{N}_{\textrm{sv}}^3) \sqrt{n_i} \log(\tfrac{1}{\epsilon})) \\ & = \tilde{\mathcal{O}}((\mathrm{c}_i^3 + {\mathcal{N}_{\textrm{sv}}^2} n_i^2 + {\mathcal{N}_{\textrm{sv}}} n_i^3 + \mathcal{N}_{\textrm{sv}}^3) \sqrt{n_i}),\\
			\mathsf{C}_i^s & = \mathcal{O}(\mathrm{c}_i^2 + n_i^2 + {\mathcal{N}_{\textrm{sv}}^2 + {\mathcal{N}_{\textrm{sv}}} n_i^2}),
		\end{split}
	\end{align}
	where $ \mathsf{C}_i^t$ and $ \mathsf{C}_i^s$ denote the time complexity and the memory complexity for each subsystem, respectively, $\epsilon$ is the solution accuracy (often on the order of $10^{-8}$), and $\tilde{\mathcal{O}}(\cdot)$ denotes $\mathcal{O}(\cdot)$ with the logarithmic factor $\log(\tfrac{1}{\epsilon})$ omitted for brevity. For simplicity, we consider the common regime in which $\mathcal T > z_i \gg n_i$. Accordingly, considering~\eqref{eq:3}, $(\mathrm{c}_i^3 + {\mathcal{N}_{\textrm{sv}}^2} n_i^2 + {\mathcal{N}_{\textrm{sv}}} n_i^3 + \mathcal{N}_{\textrm{sv}}^3) \sqrt{n_i} \asymp \mathcal{T}^3 z_i^3 \sqrt{n_i}$ and $\mathrm{c}_i^2 + n_i^2 + {\mathcal{N}_{\textrm{sv}}^2 + {\mathcal{N}_{\textrm{sv}}} n_i^2} \asymp \mathcal{T}^2 z_i^2$. Hence, we roughly obtain $\mathsf{C}_i^t = \tilde{\mathcal{O}}(\mathcal{T}^3 z_i^3 \sqrt{n_i})$ and $\mathsf{C}_i^s = \mathcal{O}(\mathcal{T}^2 z_i^2) $. To further facilitate the analysis, let us assume $\mathcal T \asymp z_i$, which leads to $\mathsf{C}_i^t = \tilde{\mathcal{O}}(z_i^6 \sqrt{n_i})$ and $\mathsf{C}_i^s = \mathcal{O}(z_i^4)$. 
	Following the same rationale, one roughly has
	\begin{align}
		\begin{split}\label{eq:new_comp}
				\mathsf{C}^t & = \tilde{\mathcal{O}} (N^6 z_i^6 \sqrt{Nn_i}) = \tilde{\mathcal{O}} (N^{6.5} z_i^6 \sqrt{n_i}),\\
				\mathsf{C}^s & = \mathcal{O}(N^4 z_i^4),
		\end{split}
	\end{align}
	where $\mathsf{C}^t$ and $\mathsf{C}^s$ denote the approximate time and memory complexity of the centralized approach, respectively.

 Comparing~\eqref{eq:3} and~\eqref{eq:new_comp}, it is now evident that our compositional approach significantly alleviates the scalability challenges associated with the centralized approach. In particular, if the centralized approach were adopted, memory complexity would have quartic dependence on the number of subsystems, potentially rendering it impractical for large-scale interconnected networks.
    In contrast, under the proposed compositional approach, the SDP~\eqref{eq: thm-iss} is solved independently for each subsystem and thus needs to be solved $N$ times, corresponding to the number of subsystems. If all subsystems are solved fully in parallel, the aggregate memory complexity grows linearly with $N$, whereas if they are solved sequentially, the aggregate time complexity grows linearly with $N$. Thus, we consider that the aggregate complexity of the SDP-solution stage of our proposed framework scales linearly with $N$.
	It is also worth noting that, by analogy with the practical dimension limit reported in~\cite{zhang2018efficient} and in view of $\mathsf{C}_i^s$ in~\eqref{eq:3}, solving the SDP~\eqref{eq: thm-iss} in our compositional framework may remain computationally tractable when $z_i$ is on the order of a few hundred. In contrast, under the centralized approach, this tractability limit becomes much more restrictive, since the memory complexity grows \emph{quartically} with the number of subsystems.

 For illustration, Fig.~\ref{fig:scalability} depicts both the time and {memory} complexities as functions of the number of subsystems for a homogeneous network of scalar linear systems (\emph{i.e.}, $z_i = 1$). As shown, even for this simple example, the centralized approach suffers drastically from scalability challenges, which become even more severe in more complex cases.

\begin{figure}[t!]
	\begin{subfigure}{0.225\textwidth}
		\centering
		\includegraphics[width=1.03\linewidth]{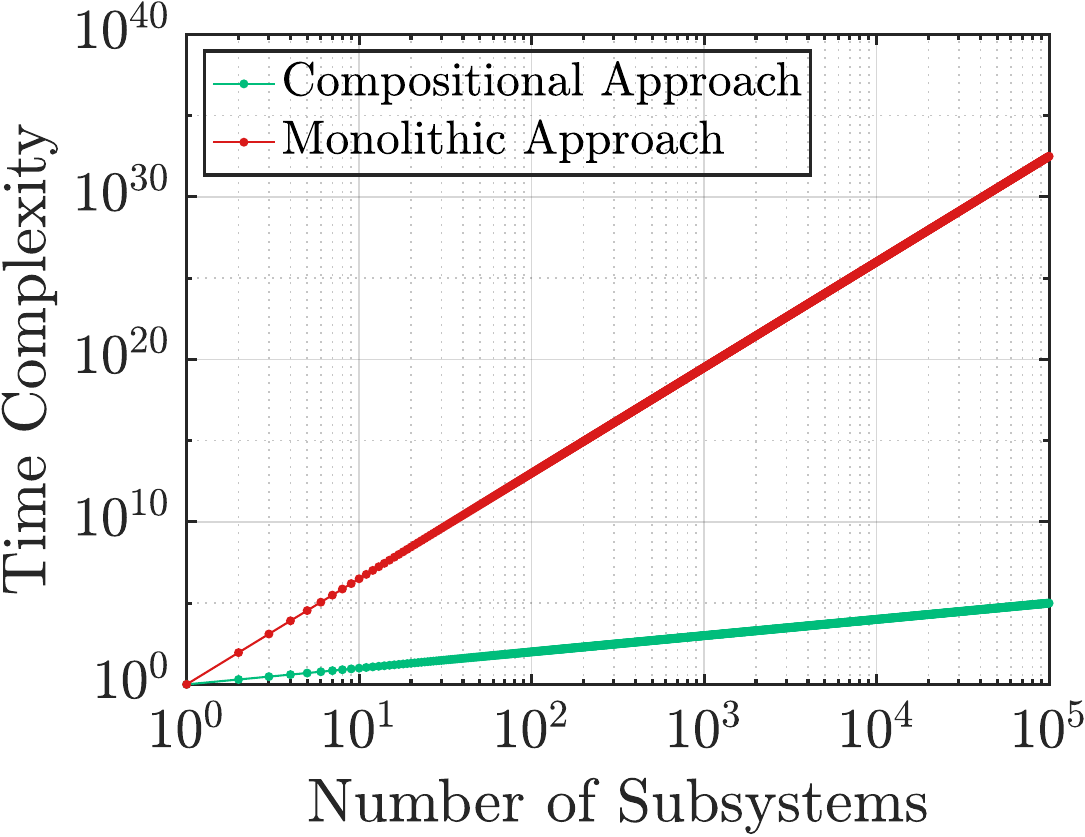}
		\caption{Time complexity}
		\label{fig:scalability2}
	\end{subfigure}\hspace{0.2cm}
	\hfil
	\centering
	\begin{subfigure}{0.225\textwidth}
		\centering
		\includegraphics[width=1.03\linewidth]{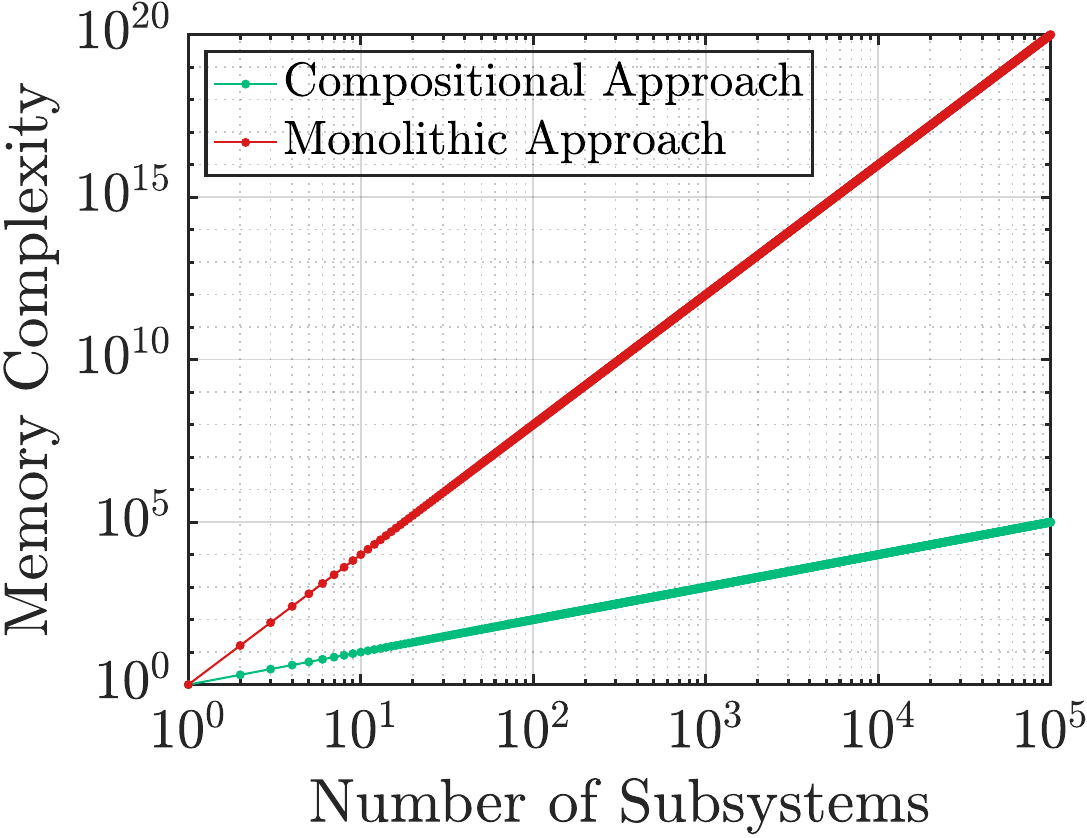}
		\caption{Memory complexity}
		\label{fig:scalability1}
	\end{subfigure}
	\caption{Comparison of the asymptotic time and {memory} complexities between the centralized and compositional methods as the number of subsystems $N$ increases.}
	\label{fig:scalability}
\end{figure}

\subsection{Conservatism Analysis}
 Here, we focus on the potential conservatism introduced by the compositional approach compared with its centralized counterpart. Specifically, in the centralized case, if it is computationally tractable to solve the problem, the network is treated as a single system; thus, there is no need to decouple the influence of neighboring subsystems on one another. Stability of the overall network is then studied as a whole, and consequently, no additional decomposition-induced conservatism is introduced. In contrast, in our compositional framework, for each subsystem, we capture the influence of neighboring subsystems via the internal input $w_i$. Subsequently, in Theorem~\ref{Theorem-data-iss}, we design an ISS Lyapunov function and the corresponding controller for each nominal subsystem, which essentially ensures robustness with respect to the internal input $w_i$. This robustness is captured by the interaction gain $\rho_i$, which is designed as $\rho_i = \tfrac{\Vert \mathcal{D}_i \Vert^2}{\mu_i}$. This gain plays a crucial role in satisfying the compositional condition~\eqref{eq: comp con}, which is one source of conservatism in the compositional framework, while achieving a substantial reduction in the scalability challenges faced by the centralized approach. In particular, the dependence on $\mathcal{D}_i$ can increase the value of $\rho_i$ in certain cases (\emph{e.g.}, a network with a fully connected topology), making it more difficult to satisfy~\eqref{eq: comp con}, which again reflects conservatism. We emphasize that~\eqref{eq: comp con} is a \emph{sufficient} condition, not a necessary one. This means that if it holds, GAS of the nominal interconnected network is guaranteed; however, if it fails, no conclusion can be drawn regarding GAS of the nominal network.

 We finally emphasize that, if the SDP~\eqref{eq: thm-iss} and the small-gain condition~\eqref{eq: comp con} are both satisfied, the closed-loop network is GAS at the origin even in the presence of nonvanishing external perturbation \(\mathcal{E}\). This is because the data-driven ISM controller cancels the matched perturbation \(\mathcal{E}\), ensuring that the perturbed network evolves identically to its nominal counterpart (cf. Corollary~\ref{Corollary: ext}). Hence, robustness is fully preserved, and no degradation occurs when the conditions hold. The only trade-off lies in the above-mentioned conservatism, which should not be interpreted as a loss of robustness but rather as a trade-off in certification compared with a centralized approach.

\subsection{Feasibility Analysis}\label{Section_feas}
Let us now discuss the feasibility of the proposed SDP~\eqref{eq: thm-iss}. When only \(\PD\) has full row rank, Theorem~\ref{Theorem-data-iss} provides a sufficient condition, meaning that if the SDP~\eqref{eq: thm-iss} is feasible, then it returns an input-to-state stabilizing controller that also linearizes the closed-loop dynamics. We now show that the SDP~\eqref{eq: thm-iss} provides an ``if and only if'' result, provided that $\operatorname{rank}(\mathbfcal J_i)=z_i+m_i$, where $\mathbfcal{J}_i = [\PD^{\! \!\! \top} \;\; \Ii^{\! \! \top}]^\top$. More precisely, for the prescribed value of $\kappa_i>0$ used in the SDP~\eqref{eq: thm-iss}, under this rank condition, the SDP~\eqref{eq: thm-iss} is feasible whenever there exists an input-to-state stabilizing feedback law that linearizes the closed-loop dynamics and admits a quadratic ISS Lyapunov function of the form considered in Theorem~\ref{Theorem-data-iss}, with the same decay rate $\kappa_i$, satisfying~\eqref{eq: feasibility_Lyap_ineq}.
	
Since the ``if'' direction follows directly from Theorem~\ref{Theorem-data-iss}, we focus on proving the ``only if'' direction. To do so, we first recall that, according to Remark~\ref{Remark: matched_non}, $\mathcal{A}_i \coloneq \begin{bmatrix}
		\mathcal{A}_{i_1} & \mathcal{A}_{i_2}
	\end{bmatrix}$, where $\mathcal{A}_{i_1} \in \R^{n_i \times n_i}$ and $\mathcal{A}_{i_2} \in \R^{n_i \times (z_i - n_i)}$. Now, suppose that there exist a feedback gain
	\(
	\mathbfcal K_i=
	\begin{bmatrix}
		\mathbfcal K_{1_i} & \mathbfcal K_{2_i}
	\end{bmatrix}
	\)
	and a matrix \(\MPI\succ0\) such that \(u_i^\ast=\mathbfcal K_i\mathcal Z_i(x_i)\) is an input-to-state stabilizer that linearizes the closed-loop dynamics, with \(\V_i(x_i)=x_i^\top\MPI x_i\) being the corresponding quadratic ISS Lyapunov function. Since the closed-loop dynamics are linearized by hypothesis, one has
	\(
	\mathcal A_{i_2}+\mathcal B_i\mathbfcal K_{2_i} = \bbzero_{n_i\times(z_i-n_i)}.
	\)
	In addition, by the preceding assumption on the quadratic ISS Lyapunov function, for the same prescribed $\kappa_i>0$ as in the SDP~\eqref{eq: thm-iss}, there exists $\mu_i\in\Rp$ such that
	\begin{align}
		(\mathcal{A}_{i_1}+\mathcal{B}_i\mathbfcal{K}_{1_i})^\top \MPI +\MPI(\mathcal{A}_{i_1}+\mathcal{B}_i\mathbfcal{K}_{1_i}) +\mu_i\MPI^{\! 2} \preceq -\kappa_i\MPI , \label{eq: feasibility_Lyap_ineq}
	\end{align}
	where $\MPI^{\! 2} \coloneq \MPI \MPI$.
	Since \(\mathbfcal J_i\) has full row rank, for the particular gain \(\mathbfcal K_i\), there exists a matrix
	\(
	\mathcal{G}_i = \begin{bmatrix}
		\mathcal{G}_{1_i} & \mathcal{G}_{2_i}
	\end{bmatrix}\!\!,
	\)
	as introduced in Lemma~\ref{Lemma-rep}, satisfying
	\[
	\begin{bmatrix}
			\PD\\
			\Ii
	\end{bmatrix} \mathcal{G}_i = \begin{bmatrix}
		\I_{z_i}\\
		\mathbfcal{K}_i
	\end{bmatrix} \quad \Longrightarrow \quad \PD\mathcal{G}_{2_i} = \begin{bmatrix}
	\bbzero_{n_i\times(z_i-n_i)}\\
	\I_{z_i-n_i}
	\end{bmatrix}\!\!,
	\]
	which is condition~\eqref{eq : thm2}. Moreover, using the data-based equality
	\(
	\Sip-\mathcal{D}_i\Wi=\mathcal{A}_i\PD+\mathcal{B}_i\Ii,
	\)
	we have
	\(
	(\Sip-\mathcal{D}_i\Wi)\mathcal{G}_i = \mathcal{A}_i+\mathcal{B}_i\mathbfcal{K}_i,
	\)
	and, consequently,
	\[
	\begin{aligned}
		(\Sip-\mathcal{D}_i\Wi)\mathcal{G}_{1_i} & = 	\mathcal{A}_{i_1} + \mathcal{B}_i \mathbfcal{K}_{1_i},\\
		(\Sip-\mathcal{D}_i\Wi)\mathcal{G}_{2_i}  & = \mathcal{A}_{i_2} + \mathcal{B}_i \mathbfcal{K}_{2_i} = \bbzero_{n_i\times(z_i-n_i)},
	\end{aligned}
	\]
	with the latter being condition~\eqref{eq : thm1}. Now, let \(\Phi_i\coloneq\MPI^{-1}\) and define \(\mathcal Y_i\coloneq\mathcal G_{1_i}\Phi_i\). From \(\PD\mathcal G_i=\I_{z_i}\), one has
	\[
	\PD \mathcal{G}_{1_i} = \begin{bmatrix}
		\I_{n_i}\\
		\bbzero_{(z_i-n_i)\times n_i}
	\end{bmatrix} \quad \Longrightarrow \quad \PD\mathcal{Y}_i = \begin{bmatrix}
		\Phi_i\\
		\bbzero_{(z_i-n_i)\times n_i}
	\end{bmatrix}\!\!,
	\]
	which is condition~\eqref{eq : thm3}. Finally, multiplying~\eqref{eq: feasibility_Lyap_ineq} by $\Phi_i$ from the left and right, and using
	\(
	(\Sip-\mathcal{D}_i\Wi)\mathcal{G}_{1_i} = \mathcal{A}_{i_1}+\mathcal{B}_i\mathbfcal{K}_{1_i}
	\)
	and
	\(
	\mathcal{Y}_i=\mathcal{G}_{1_i}\Phi_i,
	\)
	we have
	\[
	\mathcal{Y}_i^\top(\Sip-\mathcal{D}_i\Wi)^\top +(\Sip-\mathcal{D}_i\Wi)\mathcal{Y}_i +\mu_i\I_{n_i} \preceq -\kappa_i\Phi_i,
	\]
	which is condition~\eqref{eq : thm4}. Since $\mathcal{Y}_i=\mathcal{G}_{1_i}\Phi_i$ and $\MPI=\Phi_i^{-1}$, one has
	\(
	\Ii[\mathcal{Y}_i\MPI \;\; \mathcal{G}_{2_i}] = \Ii[\mathcal{G}_{1_i} \;\; \mathcal{G}_{2_i}] = \Ii\mathcal{G}_i = \mathbfcal{K}_i,
	\)
	implying that the same feedback gain is recovered. Thus, provided that $\mathbfcal{J}_i$ has full row rank, every input-to-state stabilizing feedback law that linearizes the closed-loop dynamics, together with its corresponding quadratic ISS Lyapunov function satisfying~\eqref{eq: feasibility_Lyap_ineq}, is captured by the SDP~\eqref{eq: thm-iss}, concluding the ``only if'' direction.

\section{Simulation Results}\label{Section: Simulation}

\begin{table*}[h!]
	\centering
	\caption{A broad overview of our data-driven results across five distinct topologies of interconnected networks together with two computational cost indices \textsf{RT} and \textsf{MU}, where \textsf{RT} denotes the running time (seconds), and \textsf{MU} indicates the memory usage (MB), required for each subsystem.
		\label{table: info}}
	\resizebox{\linewidth}{!}{\begin{tabular}{@{}cccccc@{}}
			\toprule
			\multirow{2}{*}[-0.25em]{Available $\mathcal Z_i(x_i)$} & \multicolumn{3}{c}{Network parameters} & \multicolumn{2}{c}{Computation costs}\\
			\cmidrule(lr){2-4} \cmidrule(lr){5-6} \vspace{-0.25cm}\\
			{} & Topology & $N$ & $\mathcal{T}$  & \textsf{RT} (sec) & \textsf{MU} (MB)\\
			\midrule
			\multirow{5}{*}{$[x_{i_1}; x_{i_2}; x_{i_1}^2; x_{i_1} x_{i_2}; x_{i_2}^2; \sin{(x_{i_1}x_{i_2})}; {x_{i_1} \cos{(x_{i_2})}};\ln{(1 + x_{i_1}^2)}; \ln{(1 + x_{i_2}^2)}]$} 
			& Fully connected  & $1000$ & $10$  & $< 0.2$ & {$0.6$} \\
			{} & Ring   & $2000$  & $10$ & $< 0.2$ & {$0.41$}\\
			{} & Binary tree   & $4095$	&	$10$	&	$< 0.2$	& {$0.41$}\\
			{} & Star	& $2000$	& $10$	& $< 0.2$	& {$0.41$}\\
			{} & Line	& $2000$	& $10$	& $< 0.2$ & {$0.41$}\\
			\bottomrule
	\end{tabular}}
\end{table*}

In this section, we follow Algorithm~\ref{Alg:1} to demonstrate the effectiveness of our data-driven approach on a large-scale interconnected network with complex nonlinearities. Using this network as a backbone, we examine five distinct interconnection topologies (\emph{i.e.}, fully connected, ring, binary tree, star, and line) and show that our approach achieves global stabilization of these networks despite the presence of external perturbations. For the reader's convenience, Table~\ref{table: info} provides a general yet concise overview of these scenarios. All simulations were performed using \textsc{Matlab} \textit{R2023b} on a MacBook Pro (Apple M2 Max, 32 GB memory).

\subsection{Fully Connected Topology}\label{subsec:fully}
Here, we consider a fully connected network with $1000$ subsystems. Note that while we choose a homogeneous network for simplicity here, our framework can handle networks with heterogeneous subsystems. Each subsystem has two state variables $x_i = [x_{i_1}; \, x_{i_2}]$, $i \in \{1, \dots, 1000\}$, and is affected by all other subsystems in the interconnection topology. The state variables of each subsystem evolve as
\begin{align}
	\begin{split}\label{eq: ex1}
		\dot{x}_{i_1} & = x_{i_1} + x_{i_2} + 5 \times 10^{-4} \sum_{\substack{j=1 \\ j \neq i}}^{1000} x_{j_2},\\
		\dot{x}_{i_2} & = x_{i_1}^2 + x_{i_1} x_{i_2} + \ln{(1 + x_{i_2}^2)} + x_{i_1} \cos(x_{i_2}) + u_i \\ & \;\;\;\, + 5 \times 10^{-4} \sum_{\substack{j=1 \\ j \neq i}}^{1000} x_{j_1} + \underbrace{20 \sin{(100\,t)}}_{\gamma_i(t)}.
	\end{split}	
\end{align}
According to Table~\ref{table: info}, the available dictionary is
\begin{align}
	\mathcal{Z}_i(x_i) = [&x_{i_1}; x_{i_2}; x_{i_1}^2; x_{i_1} x_{i_2}; x_{i_2}^2; \sin{(x_{i_1}x_{i_2})}; x_{i_1} \cos(x_{i_2});\notag\\ &\ln{(1 + x_{i_1}^2)}; \ln{(1 + x_{i_2}^2)}].\label{New}
\end{align}
Given~\eqref{New}, one can rewrite the dynamics in~\eqref{eq: ex1} in the form of~\eqref{eq: final_subsys} with
\begin{align*}
	\mathcal{A}_i & = \begin{bmatrix}
		1 & 1 & 0 & 0 & 0 & 0 & 0 & 0 & 0\\
		0 & 0 & 1 & 1 & 0 & 0 & 1 & 0 & 1
	\end{bmatrix}\!\!, \quad \mathcal{B}_i  = \begin{bmatrix}
		0\\1
	\end{bmatrix}\! \! , \\ \mathcal{E}_i(x_i, t) & = \begin{bmatrix}
		0\\ 20 \sin{(100\, t)}
	\end{bmatrix}\!\! ,
\end{align*}
and
\begin{align*}
	\mathcal{D}_{ij} & = 5 \times 10^{-4} \times \begin{bmatrix}
		0 & 1\\1 & 0
	\end{bmatrix}\!\!, \quad \forall j\in\{1,\ldots,1000\},\;j\neq i, \\
	\mathcal{D}_i & = [\mathcal{D}_{i1}\;\ldots\;\mathcal{D}_{i(i-1)}\; \mathcal{D}_{i(i+1)}\;\ldots\;\mathcal{D}_{i\,1000}] \in \R^{2 \times 1998}.
\end{align*}
Subsequently, based on Definition~\ref{def: network}, one gets
\begin{align*}
	\mathcal{A} &= \!\! \begin{bmatrix}
		\mathcal{A}_1 & \hat{\mathcal{D}}_{1\, 2} & \hat{\mathcal{D}}_{1\, 3} & \cdots & \hat{\mathcal{D}}_{1\, 1000}\\
		\hat{\mathcal{D}}_{2\, 1} & \mathcal{A}_2 & \hat{\mathcal{D}}_{2\, 3} & \cdots & \hat{\mathcal{D}}_{2\, 1000}\\
		\vdots &  & \ddots &  &  \vdots\\
		\hat{\mathcal{D}}_{999\, 1} & \cdots & \hat{\mathcal{D}}_{999\, 998} & \mathcal{A}_{999} & \hat{\mathcal{D}}_{999\, 1000}\\
		\hat{\mathcal{D}}_{1000\, 1} & \cdots & \hat{\mathcal{D}}_{1000\, 998} & \hat{\mathcal{D}}_{1000\, 999} & \mathcal{A}_{1000}
	\end{bmatrix}\!\!,\\
	\mathcal{B} &= \mathsf{blkdiag} (\mathcal{B}_1, \dots, \mathcal{B}_{1000}), \\ \mathcal{E}(x,t) & = [\mathcal{E}_1(x_1,t); \dots; \mathcal{E}_{1000}(x_{1000},t)].
\end{align*}

We emphasize that $ \mathcal{A}_i, \mathcal{B}_i ,$ and $\mathcal{E}_i $ are all assumed to be unknown, while an extensive dictionary is available (cf. Table~\ref{table: info}). 

\begin{figure}[t!]
	\centering
	\includegraphics[width=0.7\linewidth]{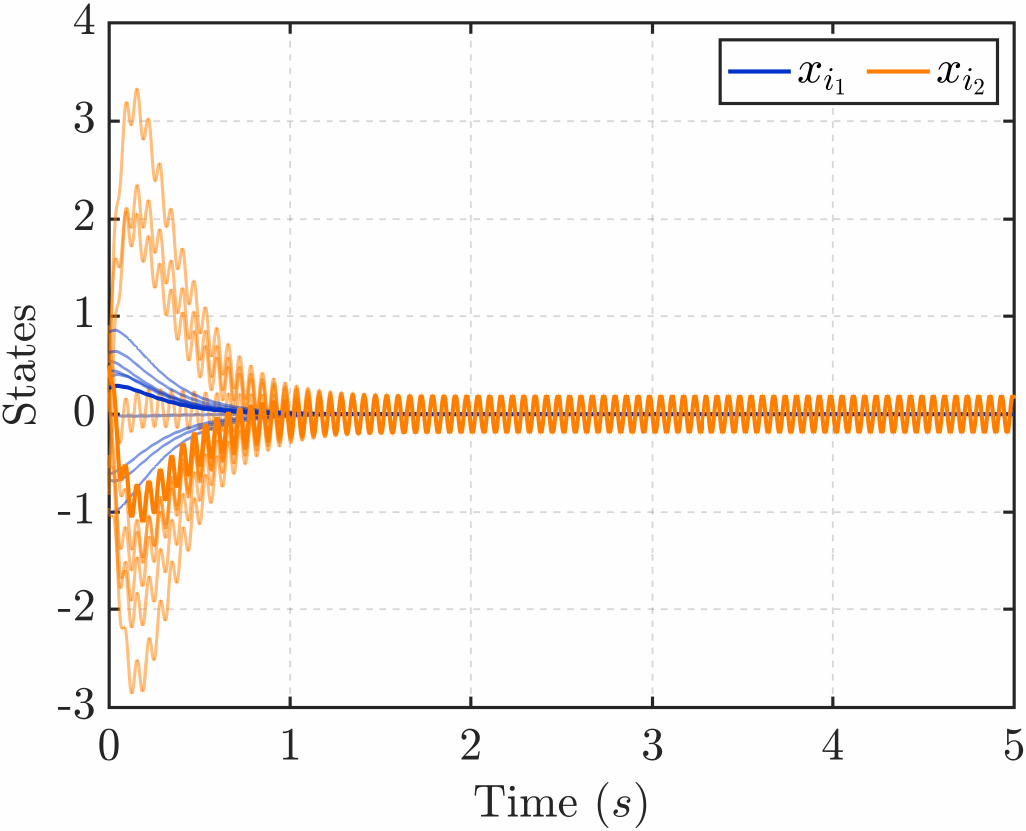}
	\caption{$10$ representative trajectories of the network (with the fully connected topology) with the designed local ISS controllers $u_i^\ast$, $i \in \{1, \dots, 1000\}$, starting from different initial conditions $x_i(0) \in [-1, \, 1]^2$, illustrating fluctuations around the origin.}
	\label{fig:ISS_Controller}
\end{figure}

To achieve global stabilization of this network, we employ our framework to design a CLF and its corresponding controller. We first collect data along a single trajectory of each nominal subsystem, as in~\eqref{eq: data_ct}, and construct $\PD$, as in~\eqref{eq: PD}. More precisely, we gather $\mathcal{T} = 10$ samples, with $\tau = 0.01$. 
To proceed, we set $\kappa_i = 2$ and $\mu_i = 1$, followed by solving the SDP~\eqref{eq: thm-iss}. Consequently, we get $\rho_i = 2.4975 \times 10^{-4}$ and
\begin{align*}
    &\mathbfcal{P}_i  = \begin{bmatrix}
        15.0903  &  3.5318\\
        3.5318   & 0.9629
    \end{bmatrix}\!\!, \; \alpha_{1_i} = 0.1292, \; \alpha_{2_i} = 15.9240,\\
    &\mathbfcal{V}_i(x_i) = 15.0903\, x_{i_1}^2 + 7.0636\, x_{i_1} x_{i_2} + 0.9629\, x_{i_2}^2,
\end{align*}
satisfying the compositional condition~\eqref{eq: comp con}.
Additionally, we obtain
\begin{align*}
    \mathbfcal{K}_i = \begin{bmatrix}
        -64.2573 &  -16.2469&  \!\! -1&  \!\! -1&  \!\! 0&  \!\! 0&   \!\! -1&   \!\! 0&  \!\! -1
    \end{bmatrix}\!\!.
\end{align*}
Having obtained $u_i^\ast$, we now design each local ISM controller $u_{i}^{\mathrm{ISM}}$. To this end, after solving \eqref{eq: lemma_eq2} and computing \eqref{eq: lemma_B}, we design $\mathcal{C}_i = [1 \;\; 1]$. We also select $\Theta_i = 20.1$ (cf. Lemma~\ref{Lemma: rho} and~\eqref{eq: Theta}). Thus, the local ISM controller $u_{i}^{\mathrm{ISM}}$ is also designed. Therefore, the CLF and the global stabilizer for the unknown interconnected network are obtained as
\begin{align*}
    \mathbfcal{V}(x) = \sum_{i=1}^{1000} \mathbfcal{V}_i(x_i), \quad u \!=\! [u_1;\dots; u_{1000}],
\end{align*}
where $u_i = u_i^\ast + u_{i}^{\mathrm{ISM}}$. Finally, we have $\alpha_1 = 0.1292$, $\alpha_2 = 15.9240$, and $\kappa = 0.0688$.

We first simulate the situation in which we solely apply $u_i^\ast$ to the subsystems without utilizing $u_{i}^{\mathrm{ISM}}$; Fig.~\ref{fig:ISS_Controller} depicts the simulation result of this scenario. As illustrated, the origin is not a GAS equilibrium point for the network, as expected. We now simulate the scenario in which $u_{i}^{\mathrm{ISM}}$, $i \in \{1, \dots, 1000\}$, are also applied together with $u_i^\ast$. We choose the initial conditions from the set $x_i(0) \in [-5\times 10^5, \, 5\times 10^5]^2$, which is adequately large to be a testament to the effectiveness of our proposed methodology. The simulation result of this scenario is depicted in Fig.~\ref{fig:ISM_Controller}. As can be seen, the trajectories are steered to the origin despite the existence of external perturbations, which is consistent with the GAS guarantee for the network. Furthermore, the time behavior of the sliding variable and its constituent terms is given in Fig.~\ref{fig:SM_Comp}, which aligns precisely with our theoretical expectations (cf. Lemma~\ref{Lemma: rho}).

\begin{remark}[\textbf{On the Fully Connected Topology}]\label{Remark: Fully}
		It should be noted that satisfying the compositional condition~\eqref{eq: comp con} for interconnected networks with a fully connected topology can often be challenging. This difficulty arises from the strong coupling among all subsystems, which is why, for this topology, we consider fewer subsystems than in other configurations (cf. Table~\ref{table: info}) and a relatively smaller magnitude for $\mathcal{D}_i$. Specifically, in a fully connected network, each subsystem is influenced by all others, implying $\mathcal{D}_{ij} \neq \bbzero_{n_i\times n_j}$, which generally increases $\Vert \mathcal{D}_i \Vert^2$, $\rho_i$ (as $\rho_i = \tfrac{\Vert \mathcal{D}_i \Vert^2}{\mu_i}$), $\hat{\rho}_{ij}$, and $\hat{\varrho}$. Condition~\eqref{eq: comp con} is easier to satisfy when $\Vert \mathcal{D}_i \Vert$ is relatively small, which occurs when the coupling matrices $\mathcal{D}_{ij}$ have small magnitudes or the number of subsystems is lower. This observation aligns with intuition in large-scale interconnected networks; stronger coupling naturally leads to a more demanding control task.
\end{remark}

\begin{figure}[t!]
	\centering
	\includegraphics[width=0.7\linewidth]{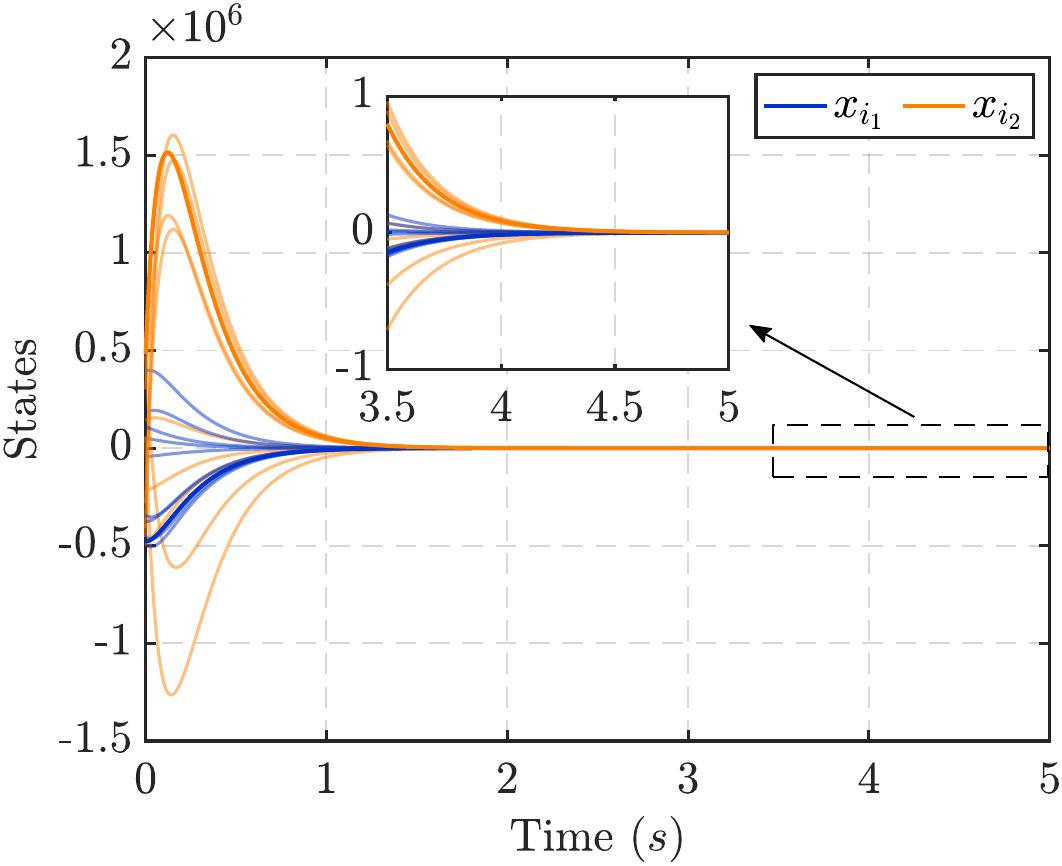}
	\caption{$10$ representative trajectories of the network (with the fully connected topology) under the designed local controllers as in Corollary~\ref{Corollary: ext}, starting from different initial conditions, where $x_i(0) \in [-5\times 10^5, \, 5\times 10^5]^2$ for all $i \in \{1, \dots, 1000\}$.}
	\label{fig:ISM_Controller}
\end{figure}

\begin{figure}[t!]
	\centering
	\includegraphics[width=0.7\linewidth]{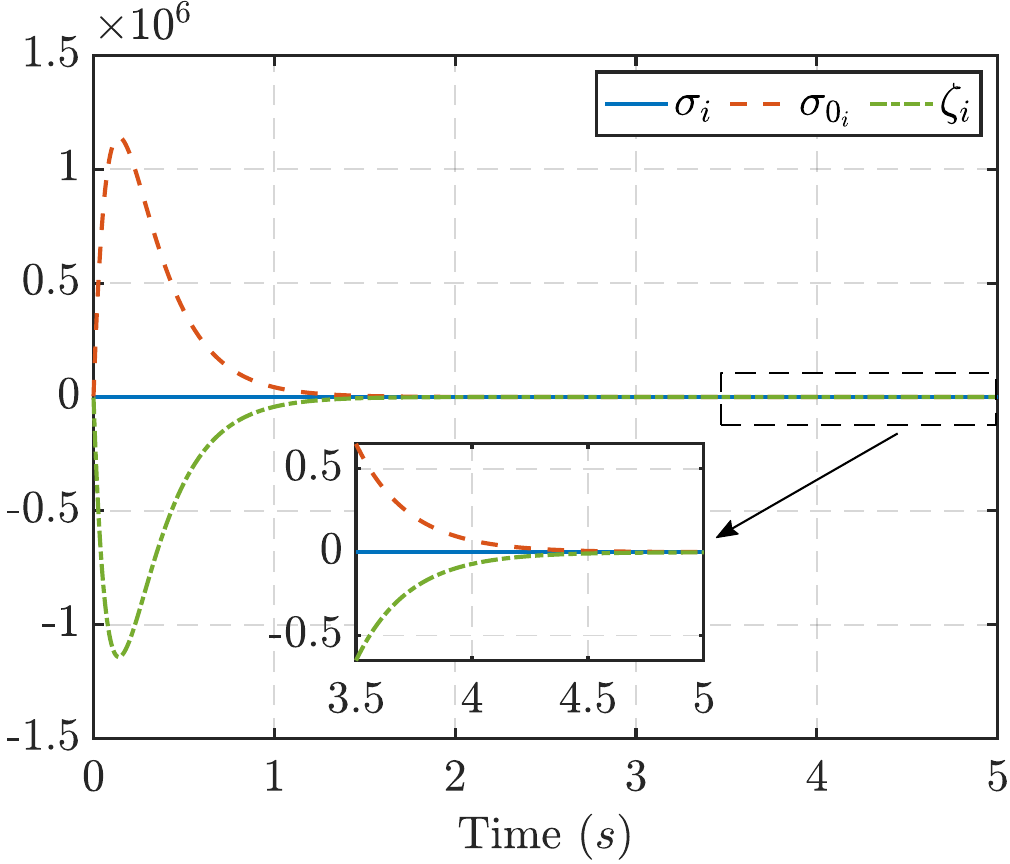}
	\caption{Representative behavior of the sliding variable and its constituent terms, shown for a single subsystem for clarity.}
	\label{fig:SM_Comp}
\end{figure}

\subsection{Ring Topology}\label{subsec:ring}
Here, we consider the network consisting of $2000$ subsystems, identical to those considered in Section~\ref{subsec:fully}, with a \emph{ring} topology, and aim to design a global stabilizer for such a network. For all $i \in \{2, \dots, 2000\}$, the dynamics of each subsystem are described as
	\begin{align}
		\begin{split}\label{eq: ex2}
			\dot{x}_{i_1} & = x_{i_1} + x_{i_2} + 0.01\, x_{(i-1)_2},\\
			\dot{x}_{i_2} & = x_{i_1}^2 + x_{i_1} x_{i_2} + \ln{(1 + x_{i_2}^2)} + x_{i_1} \cos(x_{i_2}) + u_i \\ & \;\;\;\;+ 0.01\, x_{(i-1)_1} + \underbrace{20 \sin{(100\,t)}}_{\gamma_i(t)},
		\end{split}
	\end{align}
	with the first subsystem being affected by the last one. The dictionary and the local quantities $\mathcal{A}_i$, $\mathcal{B}_i$, and $\mathcal{E}_i$ are as reported in Section~\ref{subsec:fully}, whereas $\mathcal{D}_i$ and the global matrix $\mathcal{A}$ are given by
	\begin{align*}
		\mathcal{D}_{ij} &= \begin{cases}
			\begin{bmatrix}
				0 & 0.01\\
				0.01 & 0
			\end{bmatrix} \! \! , & j=i-1,\quad i\in\{2,\ldots,2000\},\\[2mm]
			\begin{bmatrix}
				0 & 0.01\\
				0.01 & 0
			\end{bmatrix}\!\! , & (i,j)=(1,2000),\\[2mm]
			\bbzero_{2\times2}, & \text{otherwise},
		\end{cases}\\
		\mathcal{D}_i &= [\mathcal{D}_{i1}\;\ldots\;\mathcal{D}_{i(i-1)}\; \mathcal{D}_{i(i+1)}\;\ldots\;\mathcal{D}_{i\,2000}] \in\R^{2\times3998},\\
		\mathcal{A} & = \begin{bmatrix}
			\mathcal{A}_1 & \bbzero_{2 \times 9} & \bbzero_{2 \times 9} & \cdots & \hat{\mathcal{D}}_{1\, 2000}\\
			\hat{\mathcal{D}}_{2\, 1} & \mathcal{A}_2 & \bbzero_{2 \times 9} & \cdots & \bbzero_{2 \times 9}\\
			\vdots &  & \ddots &  &  \vdots\\
			\bbzero_{2 \times 9} & \cdots & \hat{\mathcal{D}}_{1999\, 1998} & \mathcal{A}_{1999} & \bbzero_{2 \times 9}\\
			\bbzero_{2 \times 9} & \cdots & \bbzero_{2 \times 9} & \hat{\mathcal{D}}_{2000\, 1999} & \mathcal{A}_{2000}
		\end{bmatrix}\!\!.
	\end{align*}

We now follow a similar procedure as in Section~\ref{subsec:fully} and set $\kappa_i = 1$ and $\mu_i = 1$. By solving the SDP \eqref{eq: thm-iss}, we obtain $\rho_i = 1 \times 10^{-4}$, and
\begin{align*}
    &\mathbfcal{P}_i  = \begin{bmatrix}
        13.0862  &  3.9115\\
        3.9115   & 1.3517
    \end{bmatrix}\!\!, \; \alpha_{1_i} = 0.1674, \; \alpha_{2_i} = 14.2705,\\
    &\mathbfcal{V}_i(x_i) = 13.0862\, x_{i_1}^2 + 7.8230\, x_{i_1} x_{i_2} + 1.3517\, x_{i_2}^2,
\end{align*}
fulfilling the compositional condition \eqref{eq: comp con}.
We also compute
\begin{align*}
    \mathbfcal{K}_i = \begin{bmatrix}
        -44.7970 &  -14.1346 &  \!\! -1&  \!\! -1&  \!\! 0&  \!\! 0&   \!\! -1&   \!\! 0&  \!\! -1
    \end{bmatrix}\!\!.
\end{align*}
The values related to the local ISM controllers remain unchanged and are the same as those reported in Section~\ref{subsec:fully}. Finally, the CLF and the GAS controller are designed as
\begin{align*}
    \mathbfcal{V}(x) = \sum_{i=1}^{2000} \mathbfcal{V}_i(x_i), \quad u \!=\! [u_1;\dots; u_{2000}],
\end{align*}
where $u_i = u_i^\ast + u_{i}^{\mathrm{ISM}},$ with $\alpha_1 = 0.1674$, $\alpha_2 = 14.2705$, and $\kappa = 0.9994$.

\subsection{Binary Tree Topology}
Here, we assess the effectiveness of our approach by applying it to an interconnected network with a binary tree topology, consisting of $4095$ subsystems.
The dynamics of each subsystem are described as
\begin{align*}
	\begin{split}
		\dot{x}_{i_1} &= x_{i_1} + x_{i_2} +10^{-2}\sum_{j=1}^{4095} \Omega(i=2j\,\vee\,i=2j+1)x_{j_2},\\
		\dot{x}_{i_2} & = x_{i_1}^2 + x_{i_1} x_{i_2} + \ln{(1 + x_{i_2}^2)} + x_{i_1} \cos(x_{i_2}) + u_i \\ & \;\;\; + \! 10^{-2}\sum_{j=1}^{4095} \Omega(i=2j\,\vee\,i=2j+1)x_{j_1} + \underbrace{20 \sin{(100\,t)}}_{\gamma_i(t)},
	\end{split}
\end{align*}
where $\vee$ stands for the logical \textsf{OR} operation and $\Omega$ is an indicator, defined as
\begin{align*}
	\Omega(i =2j\, \vee \, i = 2j+1) = \begin{cases}
		1, & i =2j\, \vee \, i = 2j+1,\\
		0, & \text{otherwise.}
	\end{cases}
\end{align*}
The dictionary and the local quantities $\mathcal{A}_i$, $\mathcal{B}_i$, and $\mathcal{E}_i$ are as reported in Section~\ref{subsec:ring}, whereas $\mathcal{D}_i$ and the global matrix $\mathcal{A}$ are given by
\begin{align*}
	\mathcal{D}_{ij} &= \begin{cases}
		\begin{bmatrix}
			0 & 0.01\\
			0.01 & 0
		\end{bmatrix}\!\!, & i=2j \,\vee\, i=2j+1,\\[2mm]
		\bbzero_{2\times2}, & \text{otherwise},
	\end{cases} \\
	\mathcal{D}_i &= [\mathcal{D}_{i1}\;\ldots\;\mathcal{D}_{i(i-1)}\; \mathcal{D}_{i(i+1)}\;\ldots\;\mathcal{D}_{i\,4095}] \in\R^{2\times8188},\\
	\mathcal{A} & = \{\mathcal{A}_{ij}\} = \begin{cases}
		\mathcal{A}_i, & i = j,\\
		\hat{\mathcal{D}}_{i\, j}, & i =2j\, \vee \, i = 2j+1,\\
		\bbzero_{2 \times 9}, & \text{otherwise.}
	\end{cases}
\end{align*}

Following a similar procedure outlined in Section~\ref{subsec:fully}, we assign $\kappa_i = 0.5$ and $\mu_i = 1.2$. Solving the SDP \eqref{eq: thm-iss} yields $\rho_i = 8.3333 \times 10^{-5}, i \geq 2$, with $\rho_1 = 0$, and
\begin{align*}
    &\mathbfcal{P}_i  = \begin{bmatrix}
        9.9540  &  3.4329\\
        3.4329   & 1.3717
    \end{bmatrix}\!\!, \; \alpha_{1_i} = 0.1676, \; \alpha_{2_i} = 11.1582, \\
    &\mathbfcal{V}_i(x_i) = 9.9540\, x_{i_1}^2 + 6.8658\, x_{i_1} x_{i_2} + 1.3717\, x_{i_2}^2,
\end{align*}
satisfying the compositional condition \eqref{eq: comp con}. 
Furthermore, we obtain
\begin{align*}
    \mathbfcal{K}_i = \begin{bmatrix}
        -34.9104 &  -12.5502 &  \!\! -1 &  \!\! -1 &  \!\! 0 &  \!\! 0 &   \!\! -1 &   \!\! 0 &  \!\! -1
    \end{bmatrix}\!.
\end{align*}
The parameters associated with the local ISM controllers remain unchanged, as specified in Section~\ref{subsec:fully}. Finally, the CLF and the GAS controller are designed as
\begin{align*}
    \mathbfcal{V}(x) = \sum_{i=1}^{4095} \mathbfcal{V}_i(x_i), \quad u = [u_1;\dots; u_{4095}],
\end{align*}
where $u_i = u_i^\ast + u_i^{\mathrm{ISM}}$, with $\alpha_1 = 0.1676$, $\alpha_2 = 11.1582$, and $\kappa = 0.4990$.

\subsection{Star Topology}
We further evaluate the efficacy of our proposed approach by constructing an interconnected network with a star topology, consisting of $2000$ subsystems. For all $i \in \{2, \dots, 2000\}$, the dynamics of each subsystem are characterized as 
	\begin{align*}
		\begin{split}
			\dot{x}_{i_1} & = x_{i_1} + x_{i_2} + 10^{-2}\, x_{1_2},\\
			\dot{x}_{i_2} & = x_{i_1}^2 + x_{i_1} x_{i_2} + \ln{(1 + x_{i_2}^2)} + x_{i_1} \cos(x_{i_2}) + u_i \\ & \;\;\;\;+ 10^{-2}\, x_{1_1} + \underbrace{20 \sin{(100\,t)}}_{\gamma_i(t)},
		\end{split}
	\end{align*}
	while the first subsystem evolves based on
	\begin{align}
		\begin{split}\label{eq: ex4}
			\dot{x}_{1_1} & = x_{1_1} + x_{1_2},\\
			\dot{x}_{1_2} & = x_{1_1}^2 + x_{1_1} x_{1_2} + \ln{(1 + x_{1_2}^2)} + x_{1_1} \cos(x_{1_2}) + u_1 \\ & \;\;\;\;+ \underbrace{20 \sin{(100\,t)}}_{\gamma_1(t)}.
		\end{split}
	\end{align}
	Moreover, we have
	\begin{align*}
		\mathcal{D}_{ij} &= \begin{cases}
			\begin{bmatrix}
				0 & 0.01\\
				0.01 & 0
			\end{bmatrix}\!\!, & j=1,\quad i\in\{2,\ldots,2000\},\\[2mm]
			\bbzero_{2\times2}, & \text{otherwise},\end{cases}\\
		\mathcal{D}_i &= [\mathcal{D}_{i1}\;\ldots\;\mathcal{D}_{i(i-1)}\;\mathcal{D}_{i(i+1)}\;\ldots\;\mathcal{D}_{i\,2000}] \in\R^{2\times3998},
	\end{align*}
	where $i,j \in \{1, \ldots, 2000\}$ and $j \neq i$.
	Therefore, one has
	\begin{align*}
		\mathcal{A} = \{\mathcal{A}_{ij}\} = \begin{cases}
			\mathcal{A}_i, & i = j,\\
			\hat{\mathcal{D}}_{i\, j}, & j=1, i\neq j,\\
			\bbzero_{2 \times 9}, & \text{otherwise.}
		\end{cases}
\end{align*}
Setting $\kappa_i = 1.5$ and $\mu_i = 0.7$, and solving the SDP~\eqref{eq: thm-iss}, we obtain $\rho_i = 1.4286 \times 10^{-4}, i \geq 2$, with $\rho_1 = 0$, along with
\begin{align*}
    &\mathbfcal{P}_i = \begin{bmatrix}
        17.0627  &  4.4805\\
        4.4805   & 1.3808
    \end{bmatrix}\!\!, \; \alpha_{1_i} = 0.1909, \; \alpha_{2_i} = 18.2526, \\
    &\mathbfcal{V}_i(x_i) = 17.0627\, x_{i_1}^2 + 8.9610\, x_{i_1} x_{i_2} + 1.3808\, x_{i_2}^2,
\end{align*}
satisfying~\eqref{eq: comp con}. Moreover, we derive
\begin{align*}
    \mathbfcal{K}_i = \begin{bmatrix}
        -50.0896 &  -14.1281 &  \!\! -1 &  \!\! -1 &  \!\! 0 &  \!\! 0 &   \!\! -1 &   \!\! 0 &  \!\! -1
    \end{bmatrix}\!.
\end{align*}
The parameters related to the local ISM controllers remain unchanged, as in Section~\ref{subsec:fully}. Finally, the CLF and the GAS controller are designed as
\begin{align*}
    \mathbfcal{V}(x) &= \sum_{i=1}^{2000} \mathbfcal{V}_i(x_i), \quad u = [u_1;\dots; u_{2000}],
\end{align*}
where $u_i = u_i^\ast + u_i^{\mathrm{ISM}}$, with the parameters $\alpha_1 = 0.1909$, $\alpha_2 = 18.2526$, and $\kappa = 0.0042$.

\begin{figure}[t!]
	
		\begin{subfigure}{0.225\textwidth}
		\centering
	\includegraphics[width=1.03\textwidth]{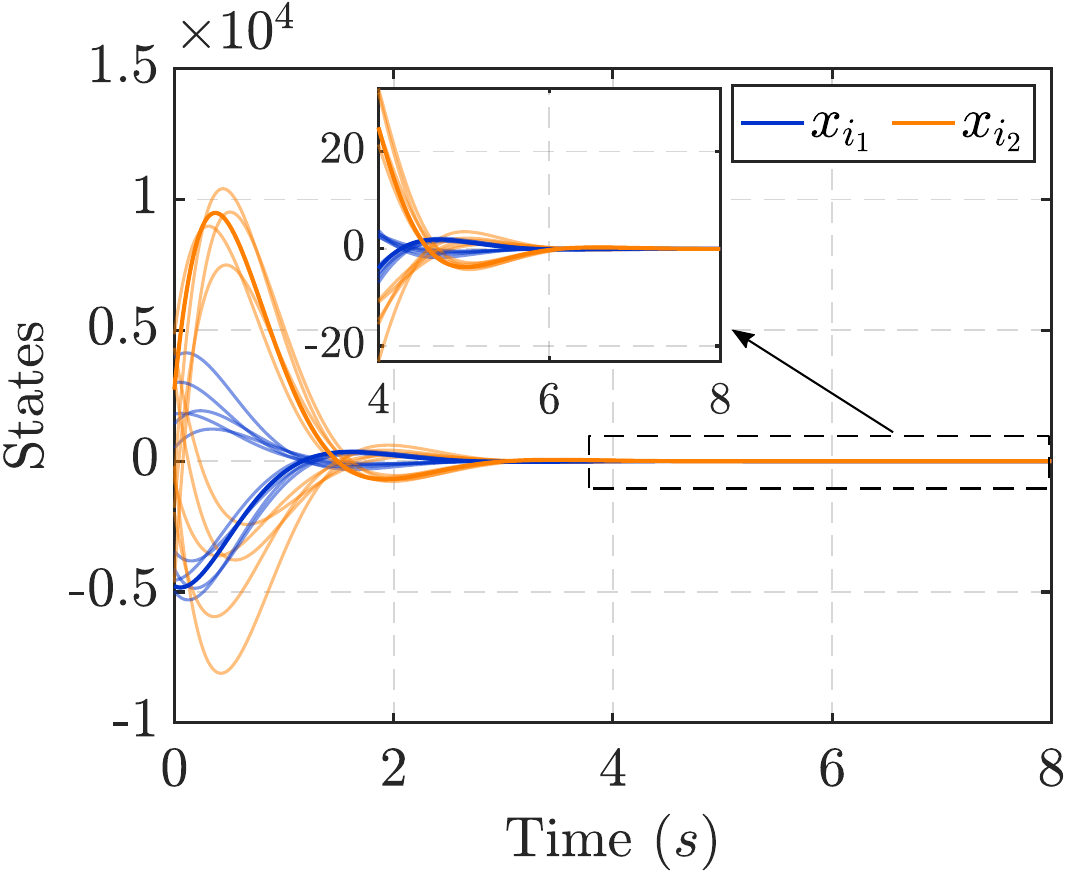}
	\caption{State variables' evolution}
	\label{fig:subfig1_1}
	\end{subfigure}\hspace{0.2cm}
	\hfil
	\centering
	\begin{subfigure}{0.225\textwidth}
		\centering
		\includegraphics[width=1.03\textwidth]{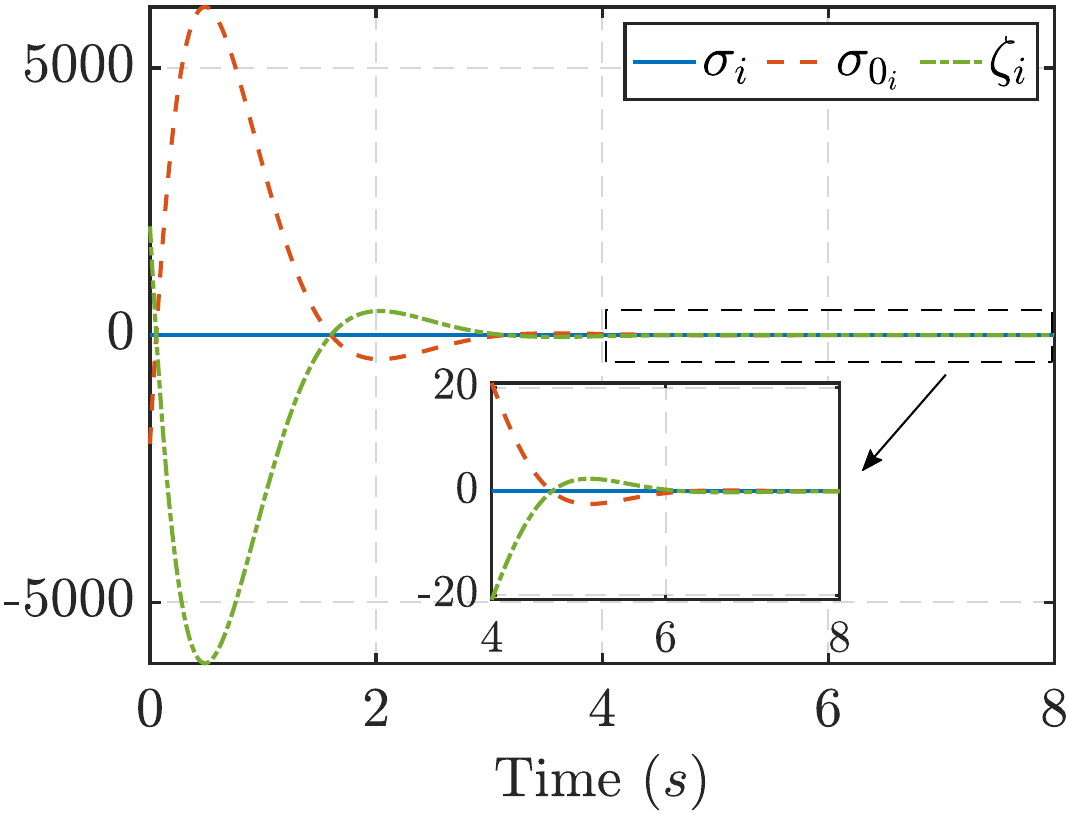}
	\caption{Sliding variable evolution}
	\label{fig:subfig2_1}
	\end{subfigure}
	\caption{(a) $10$ representative trajectories of the network (with a line topology) governed by the locally designed controllers, as described in Corollary~\ref{Corollary: ext}, with initial states \( x_i(0) \in [-5000, \, 5000]^2 \) for all $i \in \{1, \ldots, 2000\}$. (b) Representative behavior of the sliding variable and its constituent terms.
	}
	\label{fig:2}
\end{figure}

\subsection{Line Topology}
We conclude this section by assessing our data-driven framework for an interconnected network with a line topology, which comprises $2000$ subsystems.
For all $i \in \{2, \dots, 2000\}$, the dynamics of each subsystem are similar to~\eqref{eq: ex2}, whereas the first subsystem's dynamics are similar to~\eqref{eq: ex4}. Hence, we get
	\begin{align*}
		\mathcal{A} = \begin{bmatrix}
			\mathcal{A}_1 & \bbzero_{2 \times 9} & \bbzero_{2 \times 9} & \cdots & \bbzero_{2 \times 9}\\
			\hat{\mathcal{D}}_{2\, 1} & \mathcal{A}_2 & \bbzero_{2 \times 9} & \cdots & \bbzero_{2 \times 9}\\
			\vdots &  & \ddots &  &  \vdots\\
			\bbzero_{2 \times 9} & \cdots & \hat{\mathcal{D}}_{1999\, 1998} & \mathcal{A}_{1999} & \bbzero_{2 \times 9}\\
			\bbzero_{2 \times 9} & \cdots & \bbzero_{2 \times 9} & \hat{\mathcal{D}}_{2000\, 1999} & \mathcal{A}_{2000}
		\end{bmatrix}\!\!.
\end{align*}

Following the procedure of our framework, we first set $\kappa_i = 0.33$ and $\mu_i = 0.15$. Solving the SDP \eqref{eq: thm-iss} results in $\rho_i = 6.6667 \times 10^{-4}, i \geq 2$, with $\rho_1 = 0$, together with
\begin{align*}
    &\mathbfcal{P}_i = \begin{bmatrix}
        8.3454  &  2.9286\\
        2.9286   & 1.5218
    \end{bmatrix}\!\!, \; \alpha_{1_i} = 0.4373, \; \alpha_{2_i} = 9.4299, \\
    &\mathbfcal{V}_i(x_i) = 8.3454\, x_{i_1}^2 + 5.8572\, x_{i_1} x_{i_2} + 1.5218\, x_{i_2}^2,
\end{align*}
which fulfill the compositional condition \eqref{eq: comp con}. We also compute
\begin{align*}
    \mathbfcal{K}_i = \begin{bmatrix}
        -11.3491 &  -4.3779 &  \!\! -1 &  \!\! -1 &  \!\! 0 &  \!\! 0 &   \!\! -1 &   \!\! 0 &  \!\! -1
    \end{bmatrix}\!.
\end{align*}
The parameters associated with the local ISM controllers are aligned with those specified in Section~\ref{subsec:fully}. The CLF and the GAS controller can be determined as
\begin{align*}
    \mathbfcal{V}(x) &= \sum_{i=1}^{2000} \mathbfcal{V}_i(x_i), \quad u = [u_1;\dots; u_{2000}],
\end{align*}
where $u_i = u_i^\ast + u_i^{\mathrm{ISM}}$, with the parameters $\alpha_1 = 0.4373$, $\alpha_2 = 9.4299$, and $\kappa = 0.3285$. The simulation result of this topology is given in Fig.~\ref{fig:2}, which is exactly in accordance with our theoretical expectations.

\section{Conclusion}\label{Section: Conclusion}
In this work, we developed a compositional data-driven approach to achieve global stabilization of large-scale nonlinear networks with unknown mathematical models subject to external disturbances. We designed an input-to-state stabilizing controller for each nominal subsystem, along with its corresponding ISS Lyapunov function, using collected data. To achieve this, we derived conditions, formulated as data-driven SDPs, that enabled the simultaneous design of local input-to-state stabilizing controllers and associated ISS Lyapunov functions. To reject the effects of external disturbances on each subsystem and consequently on the entire network, a local ISM controller was designed for each subsystem using the collected data. Under small-gain compositional conditions, the designed data-driven ISS Lyapunov functions were utilized to construct a CLF for the network, thereby guaranteeing GAS of the nominal closed-loop network. The compositional results were then extended to perturbed networks, where we demonstrated that the synthesized ISM controllers retained GAS of the closed-loop network even in the presence of disturbances. A key direction for future work is to extend the proposed framework to handle noise-corrupted data from non-nominal subsystems and to address subsystems subject to \emph{unmatched} perturbations and nonlinearities.

\bibliographystyle{IEEEtran}
\bibliography{biblio}

\newpage 

\begin{IEEEbiography}[{\includegraphics[width=1in,height=1.3in,clip,keepaspectratio]{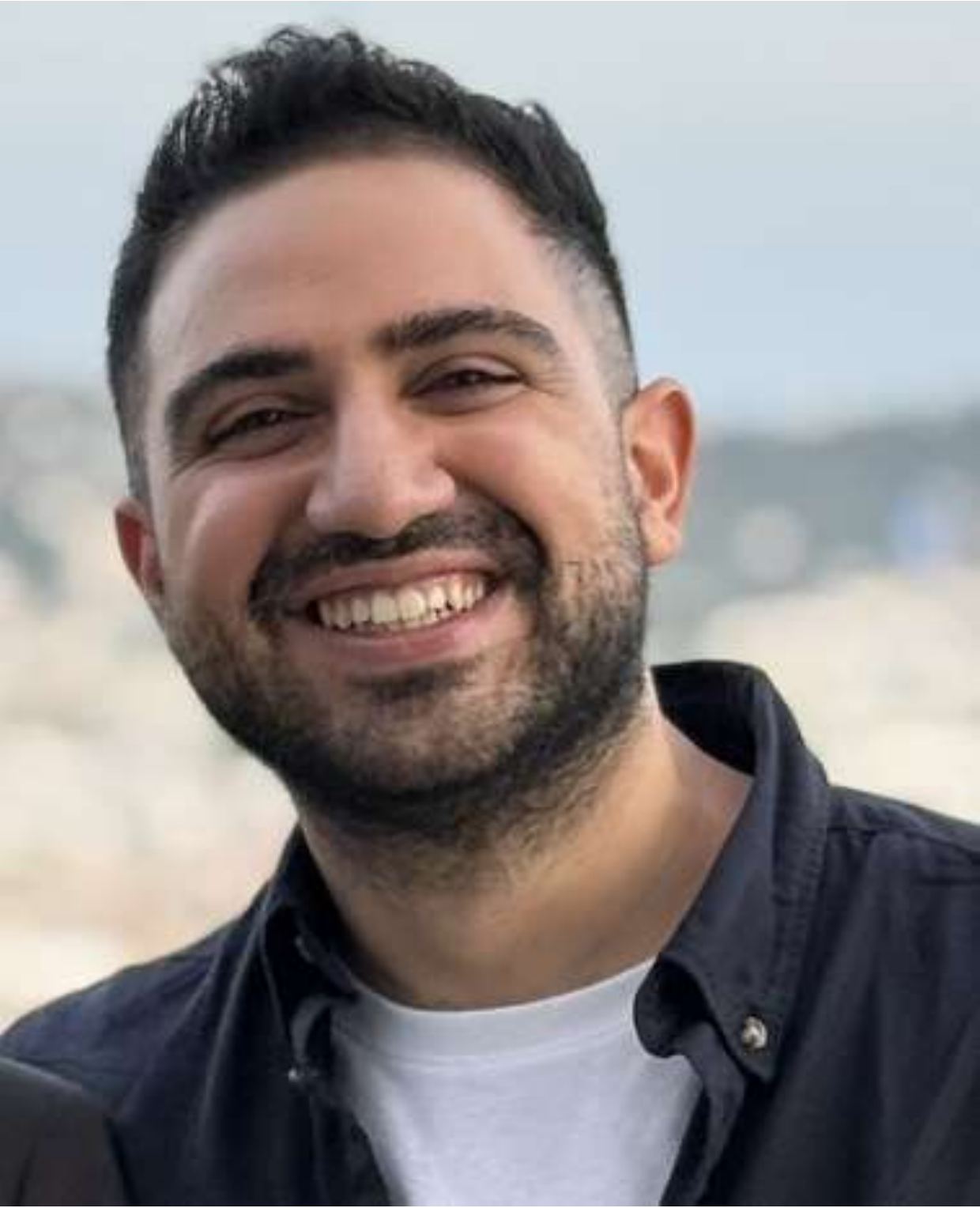}}]{Behrad Samari}~(Student Member, IEEE) received his B.Sc. and M.Sc. degrees in electrical engineering, control major, from K. N. Toosi University of Technology, Tehran, Iran, and University of Tehran (UT), Tehran, Iran, in 2019 and 2022, respectively. He is currently working toward the Ph.D. degree in computer science at the School of Computing, Newcastle University, Newcastle upon Tyne, U.K. He is the Best Repeatability Prize Finalist at the 8$^{\text{th}}$ IFAC Conference on Analysis and Design of Hybrid Systems (ADHS), 2024. His research interests include systems and control theory, data-driven approaches, and formal methods.
\end{IEEEbiography}

\begin{IEEEbiography}[{\includegraphics[width=1in,height=1.3in,clip,keepaspectratio]{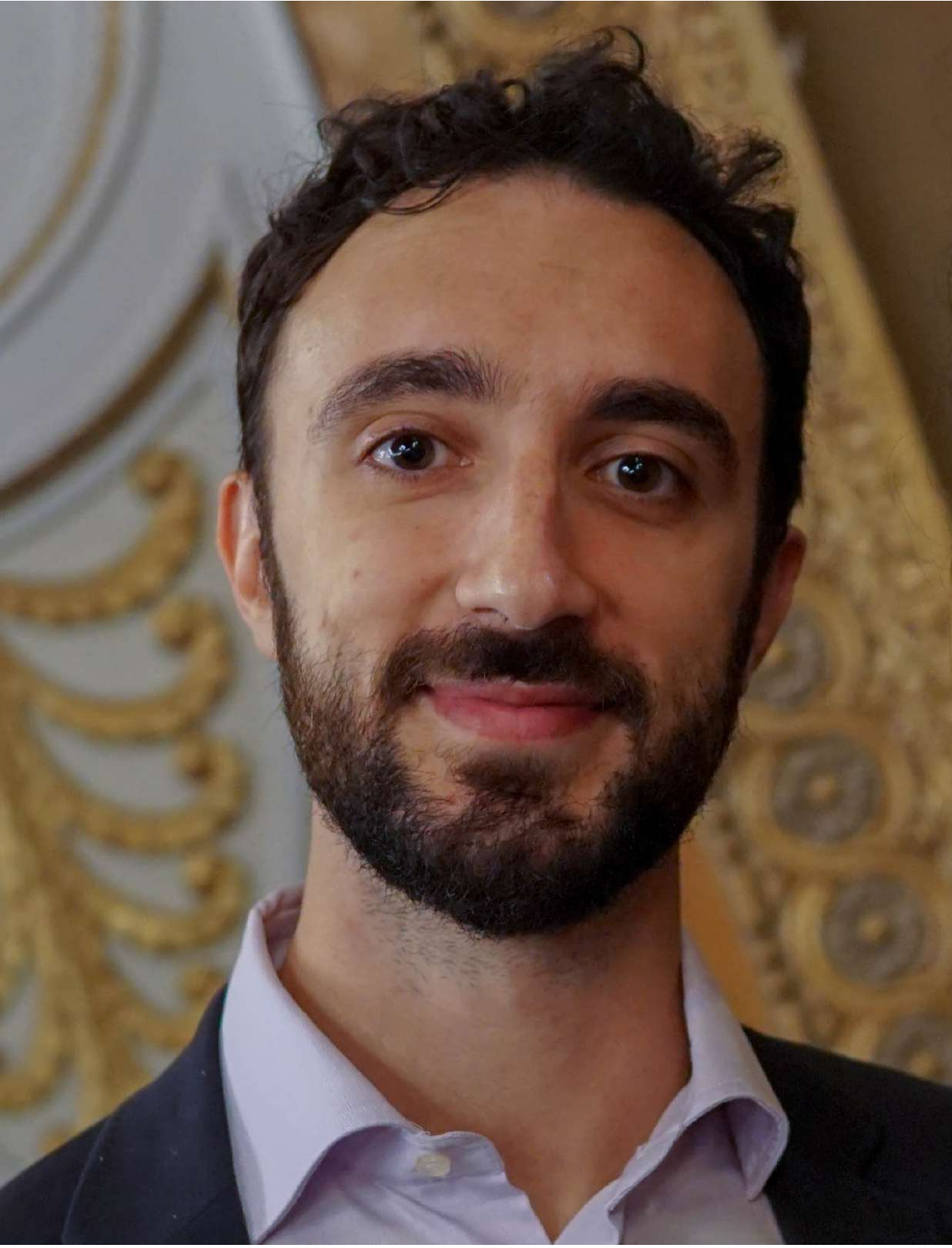}}]{Gian Paolo Incremona}
(M’10, SM’23) is an associate professor of automatic control at Politecnico di Milano. He was a student of the Almo Collegio Borromeo of Pavia, and of the Institute for Advanced Studies IUSS of Pavia. He received the bachelor’s and master’s degrees summa cum laude in Electric Engineering, and the Ph.D. degree in Electronics, Electric and Computer Engineering from the University of Pavia in 2010, 2012 and 2016, respectively. From October to December 2014, he was with the Dynamics and Control Group at the Eindhoven Technology University, The Netherlands. He was a recipient of the 2018 Best Young Author Paper Award from the Italian Chapter of the IEEE Control Systems Society, and since 2018 he has been a member of the conference editorial boards of the IEEE Control System Society and of the European Control Association. At present, he is an Associate Editor of the Journal Nonlinear Analysis: Hybrid Systems, International Journal of Control, and IEEE Control Systems Letters. His research is focused on sliding mode control, model predictive control and switched systems with applications mainly to train control, robotics and power plants.
\end{IEEEbiography}

\begin{IEEEbiography}[{\includegraphics[width=1in,height=1.3in,clip,keepaspectratio]{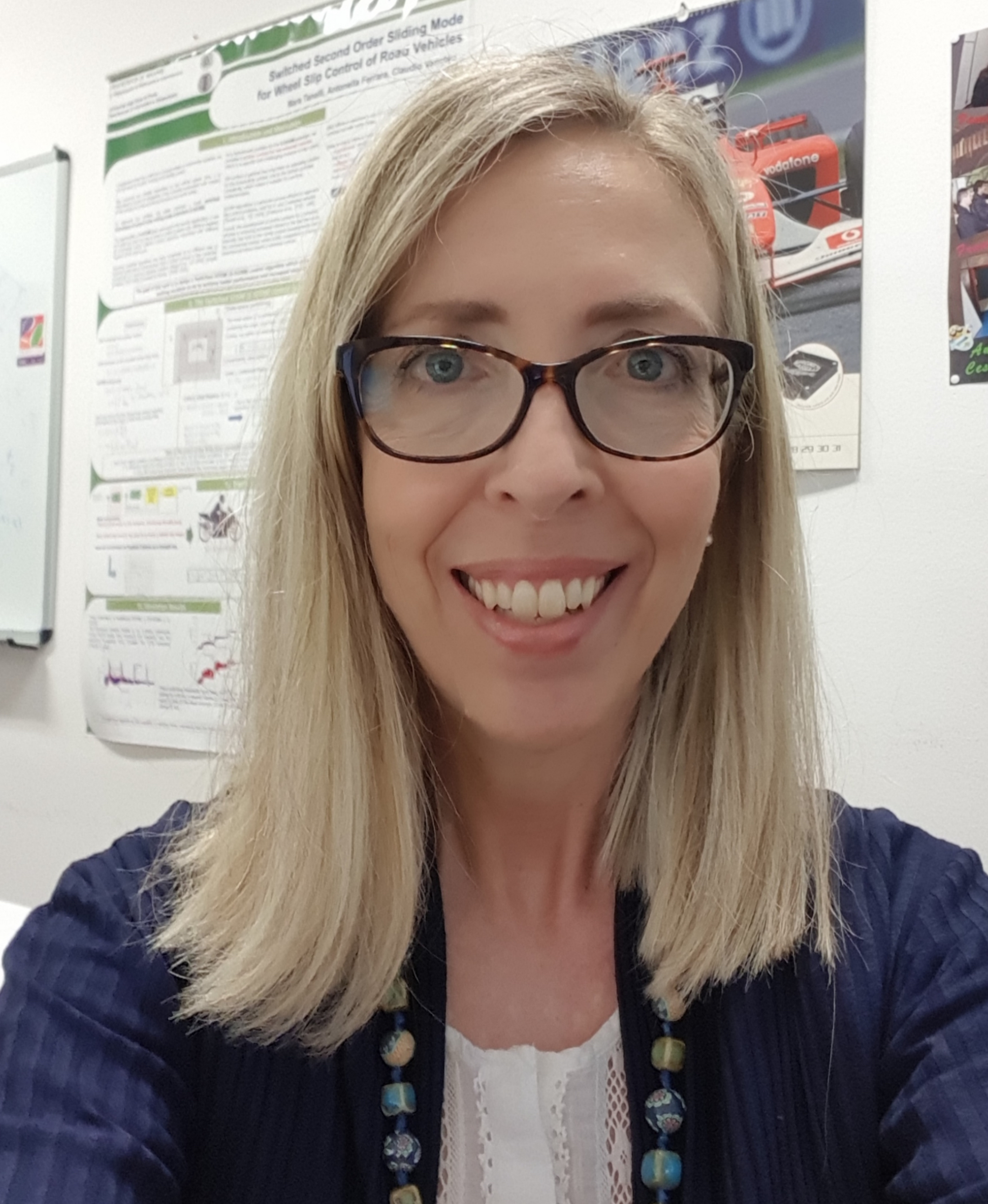}}]{Antonella Ferrara}~(SM'03--Fellow'20) received the M.Sc. degree in Electronic Engineering and the Ph.D. degree in Electronic Engineering and Computer Science from the University of Genoa, Italy, in 1987 and 1992, respectively. Since 2005, she has been a Full Professor of Automatic Control at the University of Pavia, Italy. Her research activities are mainly in the area of nonlinear control, with a special emphasis on sliding mode control, and application to robotics, power systems and road traffic. She is the author and co-author of more than 450 publications including more than 170 journal papers, 2 monographs (published by Springer Nature and SIAM, respectively) and one edited book (IET). She is currently serving as Associate Editor of Automatica, and Senior Editor of the IEEE Open Journal of Intelligent Transportation Systems. She served as Senior Editor of the IEEE Transactions on Intelligent Vehicles, as well as Associate Editor of the IEEE Transactions on Control Systems Technology, IEEE Transactions on Automatic Control, IEEE Control Systems Magazine and International Journal of Robust and Nonlinear Control. Antonella Ferrara is the Chair of the EUCA Conference Editorial Board, the Director of Operations of the IEEE Control Systems Society, the Vice-Chair for Industry of the IFAC TC on Nonlinear Control Systems (2024-2026), a member of the IFAC Industry Board and of the IFAC Conference Board.  Among several awards, she was a co-recipient of the 2020 IEEE Transactions on Control Systems Technology Outstanding Paper Award. She is a Fellow of IEEE, Fellow of IFAC, Fellow of the European Academy of Sciences (EurASc) and Fellow of AAIA. She is also a Senior Fellow of the Brussels Institute for Advanced Studies (BrIAS).
\end{IEEEbiography}

\begin{IEEEbiography}
	[{\includegraphics[width=1in,height=1.25in,clip,keepaspectratio]{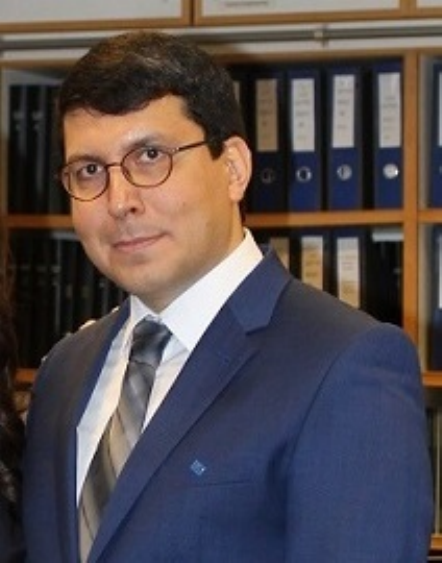}}]{Abolfazl Lavaei}~(M'17--SM'22)
	is an Assistant Professor in the School of Computing at Newcastle University, United Kingdom. Between January 2021 and July 2022, he was a Postdoctoral Associate in the Institute for Dynamic Systems and Control at ETH Zurich, Switzerland. He was also a Postdoctoral Researcher in the Department of Computer Science at LMU Munich, Germany, between November 2019 and January 2021. He received the Ph.D. degree in Electrical Engineering from the Technical University of Munich (TUM), Germany, in 2019. He obtained the M.Sc. degree in Aerospace Engineering with a specialization in Flight Dynamics and Control from the University of Tehran (UT), Iran, in 2014. He is the recipient of several international awards in recognition of his work including  Best Repeatability Prize (Finalist) at the ACM HSCC 2025, IFAC ADHS 2024, and IFAC ADHS 2021, HSCC Best Demo/Poster Awards 2022 and 2020, IFAC Young Author Award Finalist 2019, and Best Graduate Student Award 2014 at the University of Tehran with the full GPA (20/20). His research interests revolve around the intersection of Control Theory, Formal Methods, and Statistical Learning Theory.
\end{IEEEbiography}

\end{document}